\documentclass[11pt]{article}
\usepackage{geometry}                
\usepackage{array}
\usepackage{color}
\usepackage{graphicx}
\usepackage{amssymb}
\usepackage{epstopdf}
\usepackage[ruled]{algorithm2e}
\usepackage{multicol}
\usepackage{multirow}

\usepackage{url}

\usepackage{amsmath}
\usepackage{mathtools} 
\usepackage{blkarray}

\usepackage{mathtools}

\usepackage{amsthm}
\newtheorem{theorem}{{\bf Theorem}}[section]

\newtheorem{lemma}{{\bf Lemma}}[section]

\newtheorem{definition}{{\bf Definition}}[section]
\newtheorem{remark}{{\bf Remark}}[section]
\newtheorem{example}{{\bf Example}}[section]
\definecolor{purple}{RGB}{193,0,255}

\newcommand{\beq}{\begin{equation}}
\newcommand{\eeq}{\end{equation}}
\newcommand{\bdm}{\begin{displaymath}}
\newcommand{\edm}{\end{displaymath}}
\newcommand{\beqa}{\begin{eqnarray}}
\newcommand{\eeqa}{\end{eqnarray}}
\newcommand{\beqas}{\begin{eqnarray*}}
\newcommand{\eeqas}{\end{eqnarray*}}
\newcommand{\barr}{\begin{array}}
\newcommand{\earr}{\end{array}}
\newcommand{\bit}{\begin{itemize}}
\newcommand{\eit}{\end{itemize}}
\newcommand{\qq}[1]{\qquad \mbox{#1} \qquad}

\newcommand{\cE}{{\cal E}}

\newcommand{\cG}{{\cal G}}

\newcommand{\cN}{{\cal N}}
\newcommand{\cO}{{\cal O}}

\newcommand{\cV}{{\cal V}}
\newcommand{\cW}{{\cal W}}

\newcommand{\ba}{{\bf a}}
\newcommand{\bb}{{\bf b}}
\newcommand{\bc}{{\bf c}}

\newcommand{\be}{{\bf e}}

\newcommand{\br}{{\bf r}}

\newcommand{\bu}{{\bf u}}
\newcommand{\bv}{{\bf v}}
\newcommand{\bw}{{\bf w}}
\newcommand{\bx}{{\bf x}}
\newcommand{\by}{{\bf y}}

\newcommand{\bfo}{{\bf 1}}

\newcommand{\Ah}{\hat{A}}

\newcommand{\bigO}[1]{\ensuremath{\mathop{}\mathopen{}\mathcal{O}\mathopen{}\left(#1\right)}}

\title{Detecting highly cyclic structure with complex eigenpairs}

\author {
Christine Klymko \thanks{Center
for Applied Scientific Computing, Lawrence Livermore
National Laboratory, Livermore, CA 94550, USA 
(klymko1, sanders29@llnl.gov).} \and
Geoffrey Sanders\thanks{Center
for Applied Scientific Computing, Lawrence Livermore
National Laboratory, Livermore, CA 94550, USA 
(sanders29@llnl.gov).}
}

\date{}                                           

\begin{document}
\maketitle
\vspace{1cm}

\begin{abstract}
Many large, real-world complex networks have rich community structure that a network scientist seeks to understand.   These communities may overlap or have intricate internal structure.   Extracting communities with particular topological structure, even when they overlap with other communities, is a powerful capability that would provide novel avenues of focusing in on structure of interest.   In this work we consider extracting highly-cyclic regions of directed graphs (digraphs).   We demonstrate that embeddings derived from complex-valued eigenvectors associated with stochastic propagator eigenvalues near roots of unity are well-suited for this purpose.   We prove several fundamental theoretic results demonstrating the connection between these eigenpairs and the presence of highly-cyclic structure and we demonstrate the use of these vectors on a few real-world examples.
\end{abstract}


\section{Introduction}
Complex networks are found in many different disciplines and are used to model a wide variety of phenomena, from social interactions to biological processes to technological development \cite{Ca07,Es11,EsFoHi10,Ne03,Ne10}.  The analysis of these networks, which are, at their most basic, formed by objects (nodes/vertices) and connections (edges) can be useful in many aspects of study, from determining the structure of a network, to modeling or optimizing information flow, to determining most the ``important'' network elements.

One of the most commonly studied questions in network analysis is that of the detection and identification of communities, see \cite{EsKrSaXu96,Fo10,KiSoJe10,MaVa13,NeGi04,SeKoPi12} among many others.  Informally, a community in a complex network is a group of nodes that should be more closely associated with one another than with other nodes in the network, either because they perform similar functions within the network or because they form a cohesive group.  Perhaps the most commonly used definition of a community is that based on {\em modularity} \cite{Ne06}.  Informally, modularity measures the the number of internal and external edges among a subset of nodes in a graph and compares it to the number of such edges expected under a random graph model.  The subset of nodes has high modularity if there is a higher number of internal edges and a lower number of external than expected and, in this case, is said to form a good community.

As every network (graph) is associated with a number of matrices, linear algebra is a powerful and often used tool in network analysis in general \cite{BrEr05,KeGi11} and community detection specifically.  Many of these algorithms use the spectrum of the adjacency matrix  or Lapalcian of an undirected graph to find good partitions or other community structure, see \cite{FrKa01,Mc01,Vu14} among others.  The use of undirected graphs greatly simplifies many numerical approximation techniques due to the fact that the associated matrices are symmetric.  However, often complex networks have directed edges and when this edge direction is ignored many important facets of community structure can be lost.  In recent years, the amount of research on methods, including spectral methods, for community detection that take into account directed edges and higher order network structures (e.g. triangles) has been increasing \cite{BeVaBa13,BeGlLe15,KiSh12,KlGlKo14,LeNe08,RoQiYu15}.

There are many networks in which highly-cyclic structure plays an important role 
\cite{Es06,JaWaFi15,KiKi05,MiZiWi05,VaOlBa05}.  In networks with directed edges, ignoring edge direction 
can obscure details of cyclic structure.  In this work, we study linear algebraic techniques for mining graphs for various kinds of highly-cyclic structure, focusing more specifically on highly 3-cyclic structure (see Section \ref{sec:def} for a discussion of highly 3-cyclic structure). In the application of these techniques, directed graphs (digraphs) generally present both modeling and numerical approximation challenges. Scalable numerical approximation depends on iterative methods that apply basic linear algebra operations (e.g. matrix-vector multiply, inner product) to successively improve accuracy. However, iterative methods are typically less robust when applied to digraph mining, as the associated matrices are nonsymmetric. Applications involving nonsymmetric eigensolvers are typically thought of as less attractive, as the orthogonality of eigenvectors is not guaranteed, eigenpairs are possibly complex-valued, and the solvers are less robust in terms of producing highly accurate eigenvectors with a reasonable amount of work.   We argue that the analysis of nonsymmetric eigenpairs and application of nonsymmetric eigensolvers in data mining context is a research area of considerable interest for topological analyses of directed graphs.   Here we design a novel capability of using information in nonsymmetric eigenvectors to detect highly-cyclic regions of a digraph and demonstrate that the computation of these vectors is often reasonably efficient.  We prove theoretical results that pave the way for reliable and scalable algorithms.   We also outline several simple approximation techniques and do a preliminary study of their success on a few digraphs.

The rest of the paper is organized as follows.  Section \ref{sec:def} contains basic definitions and notation, including a discussion of what is meant by highly-cyclic regions in a directed network.  Section \ref{sec:motivation} provides a simple directed stochastic block model to demonstrate that, even in relatively simple examples, analysis of the underlying undirected network does not allow for easy identification of  highly-cyclic structure.  Results concerning the eigenvalues and eigenvectors of the row-normalized adjacency matrices of networks with global and local highly 3-cyclic structure are presented in Sections \ref{sec:global_structure} and \ref{sec:fuzzy}.  Section \ref{sec:experiments} contains experiments on a variety of generated and real world graphs, including on the graph from the motivating example in Section \ref{sec:motivation}.  Concluding remarks and discussion of future work can be found in Section \ref{sec:conclusions}.


\section{Definitions and Notation}
\label{sec:def}
A {\em directed graph} or {\em digraph} $G$ is defined as $G=(V, E)$ where $V$ is a set of $n$ vertices and $E=\{(i,j) \, | \, i, j \in V\}$ is a set of $m$ directed edges made up of ordered pairs of vertices.  The existence of $(i, j) \in E$ means that $G$ has an edge that points from {\em source vertex} $i$ to {\em target vertex} $j$.   Here, $(i, j) \in E$ does not imply $(j, i) \in E$.  Graphs where edges are formed by unordered pairs of vertices (and, thus, the implication holds) are called {\em undirected} graphs.  In a directed graph, if both $(i,j)$ and $(j,i)$ are in the edge set, they are often referred to as {\em reciprocal edges}. Each vertex $i$ has an {\em in-degree}, $d_i^{in}$, and an {\em out-degree}, $d_i^{out}$.    The in-degree counts the number of edges which terminate at vertex $i$, that is edges of the form $(j,i) \in E$.  The out-degree counts the number of edges of the form $(i,j) \in E$, which start at node $i$. In the remainder of this paper, we will use $d_i$ in place of $d_i^{out}$ for terseness. The {\em (total) degree} of node $i$ is given by $d_i^{tot}= d_i^{in} + d_i$.  In a directed graph, edges $(i,j)$ and $(j,i)$ are separate edges and contribute toward $d_i$ and $d_i^{in}$, respectively.  

A {\em walk} of length $k$ in a directed graph is sequence of $k+1$ nodes $i_1, i_2, \ldots i_k, i_{k+1}$ such that $(i_l, i_{l+1}) \in E$ for $1 \leq l \leq k$.  A {\em closed walk} of length $k$ is a walk of length $k$ where $i_1 = i_{k+1}$.  A {\em path} of length $k$ is a walk with no repeated nodes and a {\em cycle} of length $k$ is a closed path.  A (di)graph is {\em simple} if it has unweighted edges, there are no loops (edges from a node to itself), and no multiple edges.  An undirected graph is {\em connected} if there is a path between every pair of nodes.  A directed graph is connected if its underlying undirected graph is connected.  A digraph is {\em strongly connected} if there is a directed path between every pair of nodes.  Unless otherwise specified, all graphs considered in this paper are simple, strongly connected digraphs.

The {\em (directed) adjacency matrix} of $G$ is given by $A = (a_{ij})$ with 
$$a_{ij}=\left\{\begin{array}{ll}
1,& \textnormal{ if } (i,j) \in E,\\
0, & \textnormal{ else. }
\end{array}\right .
$$
The {\em out-degree matrix} of $G$ is given by $D = \mbox{diag}(A{\bf 1}) = (d_{ij})$ with 
$$d_{ij}=\left\{\begin{array}{ll}
d_i,& \textnormal{ if } i = j,\\
0, & \textnormal{ else. }
\end{array}\right .
$$
The {\em in-degree} and {\em (total) degree matrices} of $G$ can be defined similarly.  The {\em stochastic transition matrix} associated with the directed graph $G$ is given by 
$$B = D^{-1}A = (b_{ij}) = \left\{\begin{array}{ll}
\frac{1}{d_i},& \textnormal{ if } (i,j) \in E,\\
0, & \textnormal{ else. }
\end{array}\right .
$$

Clearly, $B$ is row stochastic, so the spectral radius of $B$ is given by $\sigma(B) = 1$.  If $G$ is a simple, strongly connected digraph, then $B$ is also irreducible (and vice versa).  In this case, by the Perron-Frobenius theorem \cite[p. 667]{Me00}, $\lambda_1 = 1$ is a simple eigenvalue of $B$ and both the left and right eigenvectors of $B$ associated with $\lambda_1$ can be chosen to be positive.

The {\em singular value decomposition (SVD)} of a matrix $A \in \mathbb{C}^{n \times n}$ is given by $A = U\Sigma W^*$ where $\Sigma = \rm{diag}(\sigma_1, \sigma_2, \ldots, \sigma_n)$, $\sigma_1 \geq \sigma_2 \geq \ldots \geq \sigma_n$ are the {\em singular values} of $A$, and $U, W \in \mathbb{C}^{n \times n}$ are orthogonal matrices whose columns are the {\em left and right singular vectors} of $A$, respectively \cite[p. 412]{Me00}.

A {\em purely $k$-cyclic graph} is a digraph $G = (V,E)$ in which $V$ is made up of $k$ non-intersecting groups of nodes, $V = V_0 \cup V_1 \cup \ldots \cup V_{k-1}$ such that $E = \{(i,j) | i \in V_l, j \in V_{(l+1)\mod k} \}$.  That is, edges only exist in a directed cycle across the {\em supernodes} $V_0, V_1, \ldots, V_{k-1}$.  A {\em highly $k$-cyclic graph} is a graph in which the probability that a (directed) edge will follow the $k$-cyclic structure is much higher than the probability that it will not.  A graph is {\em highly locally $k$-cyclic} if there is a subset of nodes in $V$ such that these nodes are highly $k$-cyclic.  This structure can also be defined in terms of random walks on the graph, as is done below in the 3-cyclic case.
\begin{definition}{(\sc Highly Three-Cyclic Structure)}
Given a connected digraph $G = (V,E)$, let $\cV \subset V$.   For any $i \in \cV$, consider a random walk $\cW = \{(i \rightarrow i_1), (i_1 \rightarrow i_2), \cdots, (i_{|\cW| -1} \rightarrow i)\}$ that starts at $i$, walks randomly with uniform probability over the out-ward edges, and returns to $i$.   If the probability that $|\cW| = 3$ is much higher than would be expected compared to a random edge placement, then we say $\cV$ has {\em highly 3-cyclic structure}.
\end{definition}

A {\em non-symmetric stochastic block model} with $b_r$ row blocks and $b_c$ column blocks is defined by a {\em row-indicator matrix} $Q_r = [Q_r]_{ij} \in \{0,1\}^{n \times b_r}$, a {\em column-indicator matrix} $Q_c = [Q_c]_{ij} \in \{0,1\}^{n \times b_c}$, and an {\em inter-block probability matrix} $P_0 = [P_0]_{ij} \in [0,1]^{b_r \times b_c}$.  The {\em edge probability matrix} is given by $P = Q_r  P_0 Q_c^t$, an $n \times n$ matrix with large rectangular submatrices of constant value.   To generate a graph from this model (given by an adjacency matrix $A$) one performs a Bernoulli trial for each edge $(i,j)$ with probability $P_{ij}$. Thus, $A_{ij} = 1$ whenever Uniform$([0,1]) < P_{ij}$, otherwise $A_{ij} = 0$.   In this work, we restrict ourselves to the case where the row and column blockings correspond, $b_r=b_c=:b$ and $Q_r = Q_c =: Q$.

Let $\theta_{p,q}$ be given by
$$\theta_{p,q} =  \exp\left(\frac{p}{q} 2\pi \iota \right)$$
where the complex unit $\iota$ satisfies $\iota^2 = -1$. Now, the $q$-th roots of unity, for $p=0,1,..., (q-1),$ are given by $\theta_{p,q}$.   It follows that $\theta_{p,q} = \theta_{p \mod q, q}$ and $\theta_{p_1 + p_2 , q} = \theta_{p_1,q}\theta_{p_2,q}.$


\section{A Motivating Example}
\label{sec:motivation}

Triangles have often been a structure of interest in complex networks, both directed and undirected.  One area of study has been to find areas in the network with a high number of triangles (a problem closely related to that of finding dense subgraphs), see \cite{FeKoPe99,SePiKo13} among many others.  Discovering such structure is useful for many areas of graph analysis, especially in community detection \cite{JiGaGaWa14,KlGlKo14,SeArGo11}.  However, as is often the case, it is easier to find and interpret triangle structure in undirected networks than in directed networks.  Part of the reason for this is due to the fact that in directed graphs, the differentiation between in- and out-edges leads to seven unique triangle structures, up to isomorphism \cite{SePiDuKo13}.  In this section, we are concerned with only one type of directed triangle, the three-cycle with no reciprocal edges.  More specifically, we are concerned with finding areas of {\em highly 3-cyclic structure} in directed graphs.

\begin{figure}
\centering
\includegraphics[width = 0.56\textwidth]{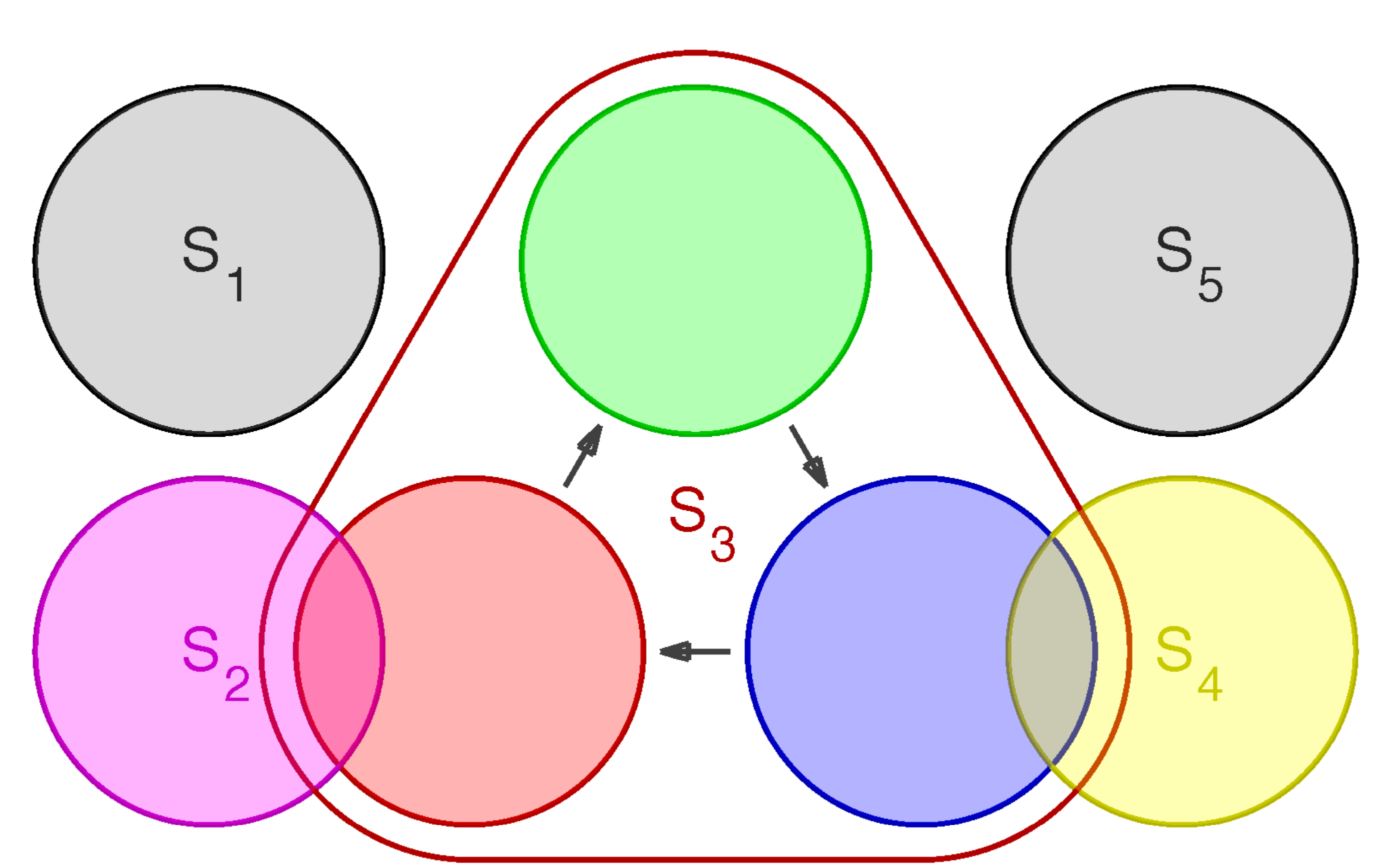} \quad
\hspace{0.1in}
\includegraphics[width = 0.33\textwidth]{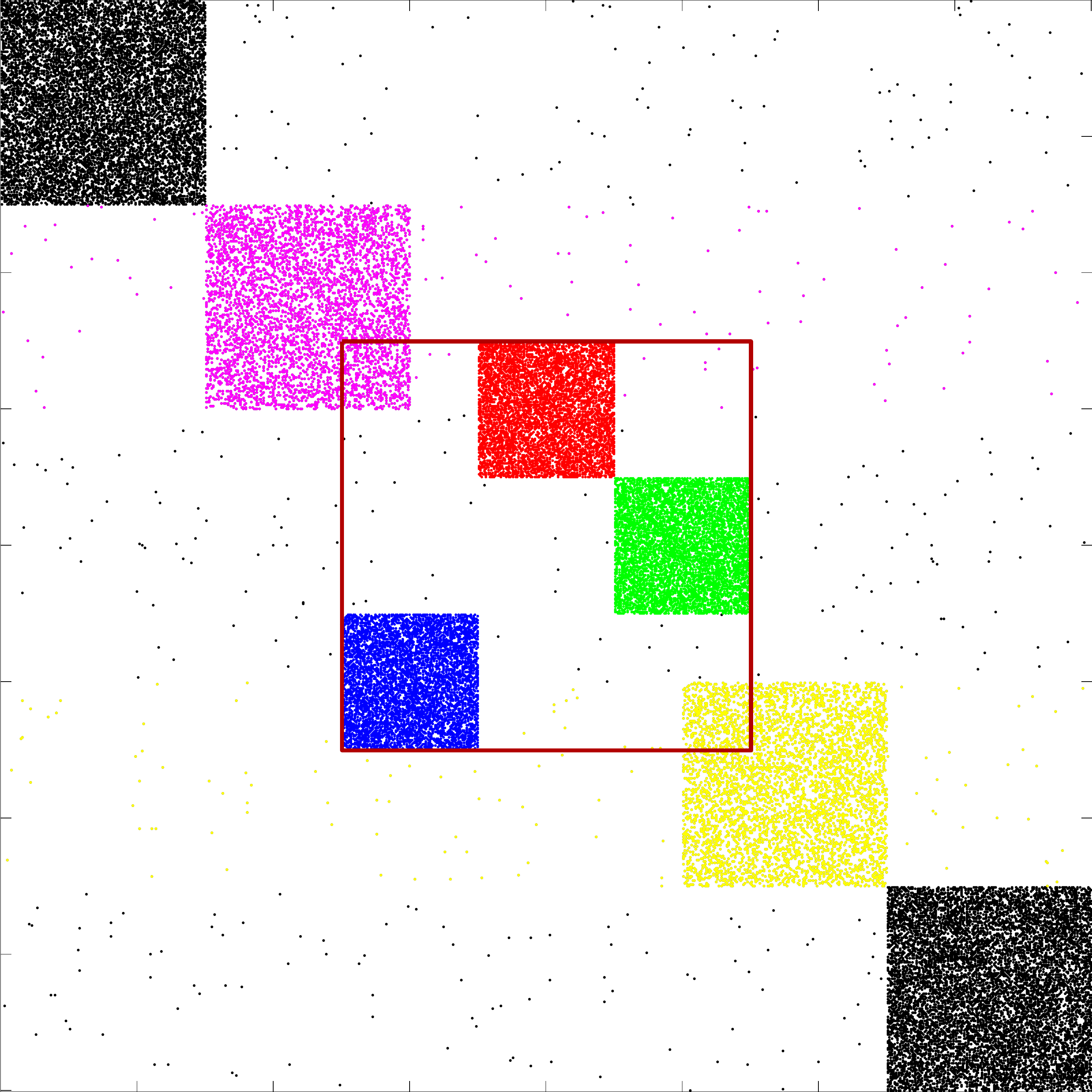}
\caption{A motivating example.   A vertex set diagram is depicted on the left.   There are four classic communities (2 black, one magenta, one yellow) and a single 3-cyclic community (dark red outline with three subsets in red, green and blue).   Two of the cyclic subsets are overlapping with two of the classic communities (red-magenta, and blue-yellow).   A matrix sparsity plot of $A$ for a realization of this setup is on the right.   Outgoing edges are represented by dots of associated color.}
\label{fig:ex0a}
\end{figure}

We consider a network generated from a non-symmetric stochastic block model which contains several dense, classical (modularity-based) communities and also a highly 3-cyclic region, some of which overlap.  The parameters of the stochastic block model used can be found in Example \ref{ex0}.  Finding a particular structure (or type of structure) of interest in a large directed graph can be quite difficult, especially when the graph contains various types of structures.  In the rest of this section, we demonstrate how spectral methods on the underlying undirected graph work well for discovering the classical, dense communities in Example \ref{ex0} but struggle to identify the highly 3-cyclic region.

\newcommand{\numext}{q_{ext}}

\begin{example}{\bf (A Hidden 3-Cyclic Community)}
\label{ex0}
Figure~\ref{fig:ex0a} depicts a particular stochastic block model that has 4 {\em classical communities} (where internal structure is purely random with constant probability) and one {\em 3-cyclic community} (where internal structure is largely dominated by edges cycling through three different subsets of vertices in order).   Two of the classical communities overlap with portions of the 3-cyclic community.   The overlap and internal structure can all be represented with a single stochastic block model with $b = 9$ blocks.   The edge probabilities are set to 0.4 for the non-overlapping classical communities, to 0.2 for the overlapping classical communities, to 0.5 for the cyclic community structure, and there is a background noise probability of 0.001.   The classical communities have 150 vertices and the cyclic community has 3 sets of 100 vertices.   The right side of Figure~\ref{fig:ex0a}  shows the sparsity structure of the adjacency matrix of a graph sampled from this model.   Additionally, we consider adding more of the non-overlapping, external communities.    Figure~\ref{fig:ex0a} depicts the case where there are 2 non-overlapping, classical (external) communities; we will analyze the cases with 8 and 14 as well.  These cases are referred to as $\numext = 2, 8,$ and 14.  We will return to this example several times throughout this paper.
\end{example}

Suppose one is given a graph such as that from Example \ref{ex0} with vertices not ordered by their community blocking, and no knowledge of the number of blocks or block sizes.    A common topological data mining goal would be to completely recover a plausible generative model (learning the blocking and all the associated probabilities).   As a general problem on stochastic block models, this endeavor is quite difficult; particularly when the number of blocks (which define classes of nodes) is quite large, the block interactivity is diverse, and blocks overlap in various ways.   Moreover, when one tries to do so with a real-world digraphs it is often the case that no simple model is plausible.

However, when one is interested in one or more specific types of graph topology, it may be unnecessary to understand the structure in the graph as a whole.   For example, in the case where we want to detect the 3-cyclic communities within a graph (e.g. $S_3$ in Figure~\ref{fig:ex0a}), the identification of communities of other types in the graph may be unimportant.  We spend the rest of this section briefly reviewing an existing SVD-based technique that can be used to recover the structure of interest for the example.   We discuss the limitations for this approach before moving on to present our complex eigenpair-based technique.


\subsection{Spectral Embedding via SVD}

\begin{figure}
\centering
\includegraphics[width = 0.318\textwidth]{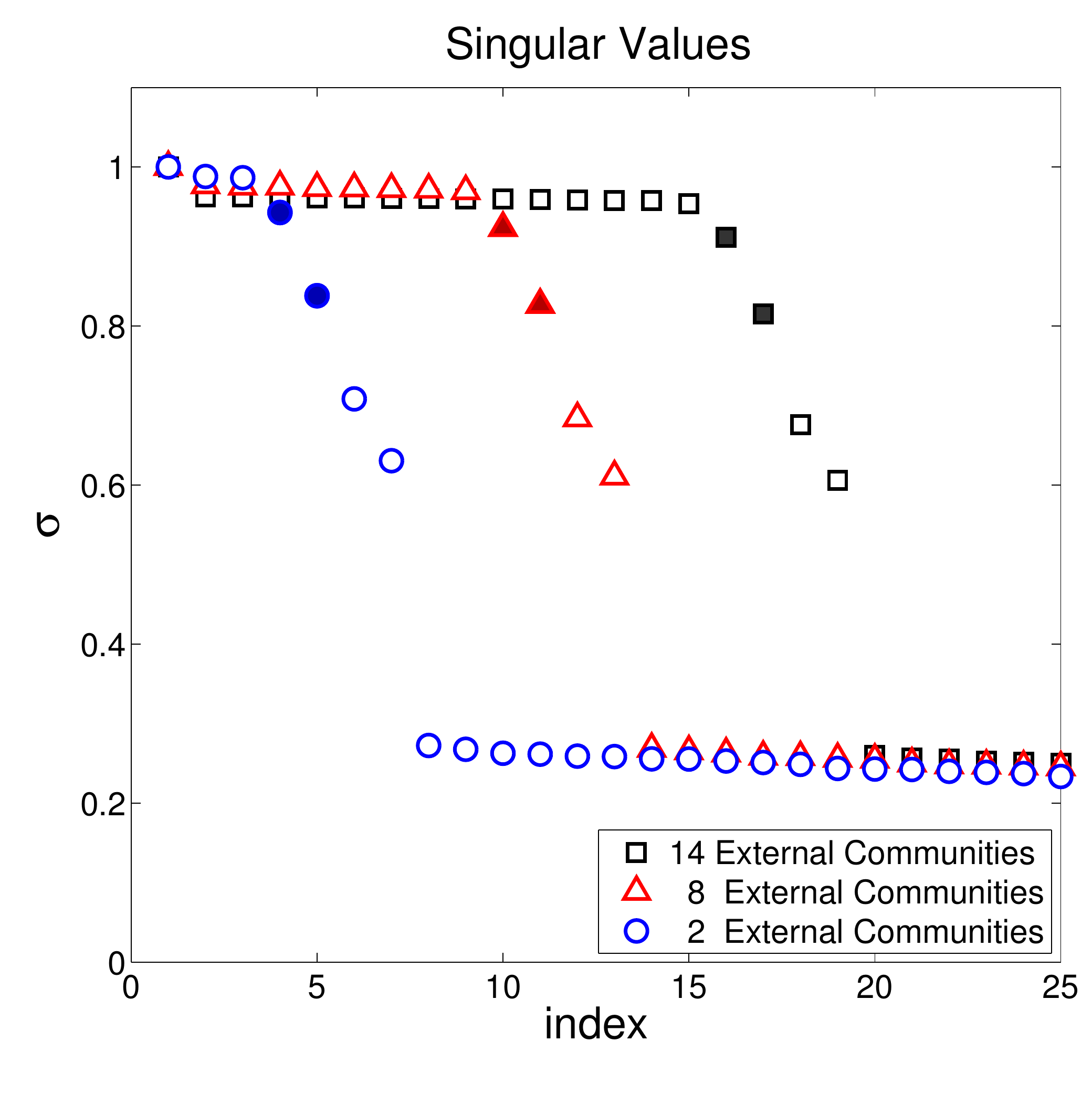} 
\hspace{0.05in}
\includegraphics[width = 0.31\textwidth]{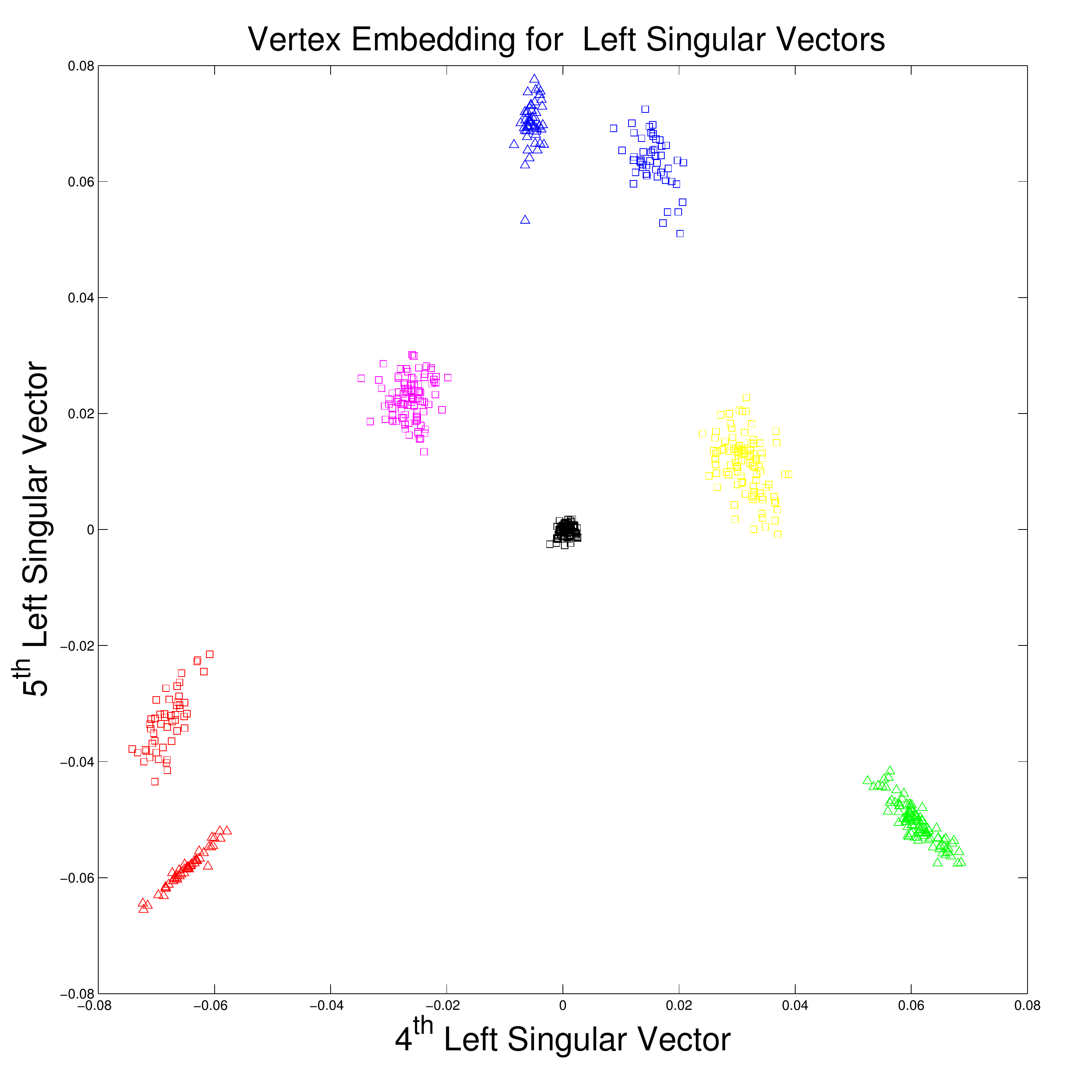}
\hspace{0.05in}
\includegraphics[width = 0.31\textwidth]{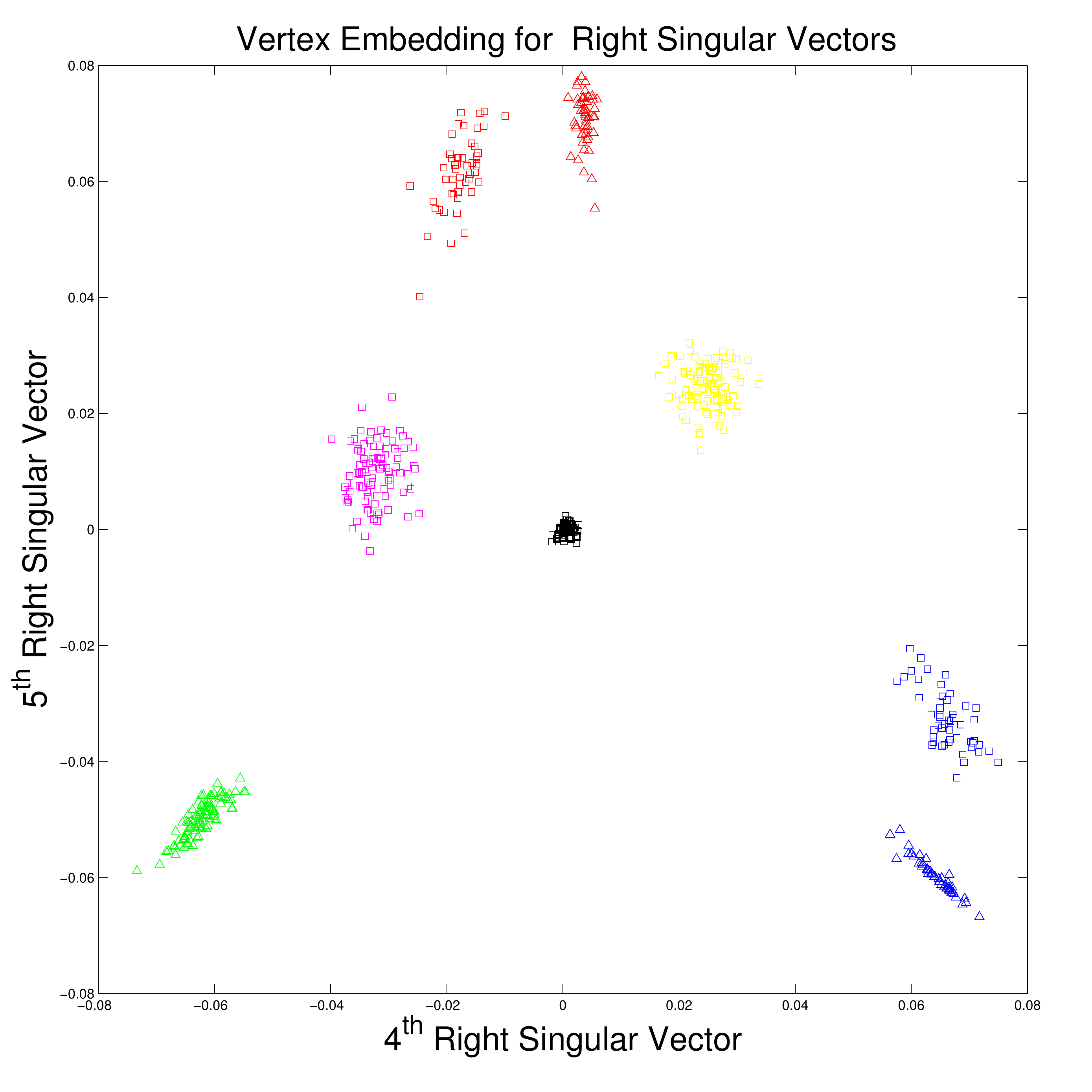} 
\caption{\small On the left are the top 25 singular values of a scaled adjacency matrix $D_r^{-1/2} A D_c^{-1/2}$ related to the graph from Example~\ref{ex0} with $ \numext = 2, 8,$ and 14 classical communities which do not overlap with the highly 3-cyclic community.  The middle plot is a two-dimensional vertex embedding from left singular vectors associated with the fourth- and fifth-largest singular values of the 2 community case (as marked in the left plot with solid blue circles).   The right plot is the respective embedding from the corresponding right singular vectors.   The solid-filled red triangles and solid-filled black squares on the left plot show the singular values for which the embeddings associated with left and right singular vectors clearly indicate the 3-cyclic structure for the cases of $\numext = 8$ and 14, respectively (not shown).   The color scheme for the spectral coordinates matches that of Figure~\ref{fig:ex0a}: in particular the three subclasses of the 3-cyclic community are red green and blue.  The red/green/blue triangles represent vertices that are internal to those classes and the red/blue squares represent vertices that overlap with the classical community structure (magenta and yellow squares).   Black coordinates are from the non-overlapping classical communities.} 
\label{fig:ex0b}
\end{figure}


For a stochastic block model with $b$ blocks, the DI-SIM algorithm \cite{RoQiYu15} uses a rank-$s$ ($s = \cO(b)$) SVD factorization to embed vertices in $s$-dimensional space, where spatial clustering algorithms are employed to find the blocking.   Given the full SVD is $A = U \Sigma W^*$, the rank-$s$ SVD is given by $U_s \Sigma_s W_s$, where matrices $U_s$ and $W_s$ are given by the first $s$ columns of $U$ and $W$ and $\Sigma_s = {\rm diag}(\sigma_1, \sigma_2, \ldots, \sigma_s)$.  These matrices are are computed and a dual spectral embedding of each vertex is available.   The coordinates for vertex $i$ are the $i$th rows of these matrices: $U_s^t \be_i$ and $W_s^t \be_i$.   A planar projection for each of these embeddings for Example~\ref{ex0} is visible in the middle and right plots of Figure~\ref{fig:ex0b}. 

Note that it is typically useful to scale rows and columns of $A$ and/or center $A$ via low-rank correction.   The results in Figure~\ref{fig:ex0b} use the SVD of $\Ah = D_r^{-1/2} A D_c^{-1/2}$, where $D_r = \mbox{diag}(A\bfo)$ and $D_c = \mbox{diag}(\bfo^t A)$ (this ensures $\|A\|_2 = 1$, as in \cite{Ro11}).

We observe that the SVD of $A$ is directly related to the spectral embedding for an undirected bipartite graph,

\beq
\label{eqn:Abp}
A_{bp} =
\left[
\barr{cc}
O & A \\
A^t & O
\earr
\right] = 
\frac{1}{2}
\left[
\barr{rr}
U &  U \\
W & -W
\earr
\right]
\left[
\barr{rr}
\Sigma &  O \\
O & -\Sigma
\earr
\right]
\left[
\barr{rr}
U &  U \\
W & -W
\earr
\right]^*.
\eeq

(Note that the SVD of the scaled matrix $\Ah = D_r^{-1/2} A D_c^{-1/2}$ can similarly be shown to related to the eigendecomposition of the normalized Laplacian matrix associated with $A_{bp}$.)  The graph associated with $A_{bp}$ has two copies of the vertex set, $\cV_r = \{1, ..., n\}$ and $\cV_c = \{n+1, ..., 2n\}$.   Each edge from the original directed graph is also rewired to connect vertices in $\cV_r$ to those in $\cV_c$: for each $(i,j) \in E$, we have $(i, n+j) = (n+j,i)$ in the undirected bipartite graph.  This connection allows us to reason about the SVD of a nonsymmetric matrix through the spectral decomposition of a symmetric matrix.  The graph associated with $A_{bp}$ may be disconnected (even if $G$ is strongly connected).   This can cause some amount of degradation of structures abundant in cycles of length 3 or longer, due to the inclusion of backwards edges; powers of $A_{bp}$ contain $A^t A$ and $AA^t$.   Example \ref{ex0a} contains an simple extreme example which demonstrates this.

\begin{example}
\label{ex0a}
Let $S_1 \in \{0,1\}^{n \times n}$ be a {\em permutation matrix} (there is exactly one $1$ in each row and each column).    For $A = S_1$, the associated graph is a union of disconnected cycles.   Yet, the bipartite graph $A^{(1)}_{bp}$ (formed as in Equation~(\ref{eqn:Abp})) is a collection of $n$ disconnected reciprocal edges on $2n$ vertices.   If $S_2$ is another permutation matrix that is not isomorphic to $S_1$, the respective $A^{(2)}_{bp}$ will still be isomorphic to $A^{(1)}_{bp}$.   All information about the number of cycles and the cycle length(s) is completely lost in an SVD factorization.   
\end{example}

Returning to Example~\ref{ex0} and $\numext =2, 8, 14$, we calculated the rank-25 SVD, and plotted the singular values in the left of  Figure~\ref{fig:ex0b}.   For $\numext=2$, we see there is a large gap between the seventh and the eighth singular value (and we have seven dominant classes of vertices).   We embed in the 7 dimensional space associated with the the left singular vectors. For this small example, {\tt dbscan} \cite{EsKrSaXu96} is easily tunable to accurately recover all 9 blocks within the 7d space.   After looking at several projections into two and three dimensional space we noticed that most of the separation of the blocks is given by the planar projection onto 
the left singular vectors associated with the fourth and fifth singular values, which are plotted in the middle of Figure~\ref{fig:ex0b}.   This projection also has the attractive feature that vertices in blocks corresponding to the highly-cyclic community are most easily separable from the rest of the vertices.   Follow on analysis can be used to determine that these blocks do compromise a single 3-cyclic community.   

In the $\numext = 8$ and 14 cases, however, we note that the separation in singular values happens at a higher index (13 and 19, respectively) and we consider embeddings into 13 and 19 dimensional space.   Using spatial clustering algorithms, it is much more difficult to correctly resolve the blocks associated with the highly-cyclic structure in these high-dimensional embeddings.   For $\numext=8$, we found that planar projections from singular vectors associated with the $10^{th}$ and $11^{th}$ largest singular values (marked with solid-red triangles on the left of Figure~\ref{fig:ex0b}) are highly useful for resolving the highly-cyclic structure (in fact the corresponding embeddings are qualitatively identical to those in Figure~\ref{fig:ex0b}, with the additional external communities also embedded near the origin).     For $\numext=8$, the $16^{th}$ and $17^{th}$ singular values (marked with solid-black squares on the left of Figure~\ref{fig:ex0b}) gave similar useful embeddings.   We used considerable knowledge of the desired structure to find these planar projections.

This example clearly demonstrates that the SVD approach can be somewhat attractive for detection of highly-cyclic structures within stochastic block models having few blocks.   The primary drawback is that one needs to compute $O(b)$ singular-value triplets to robustly resolve the structure.    Then spatial clustering must be performed in this very high-dimensional space.   Lastly, the follow on analysis is more difficult with many more blocks.   This poses severe difficulties for the scalability of the SVD approach ( e.g. if one is looking for a small number of highly 3-cyclic structures in a large graph with thousands of classical communities, one may have to dig fairly deep into the SVD to pull the blocks out).   This does not indicate that information in the SVD cannot ever be efficiently used to detect these structures when the number of blocks is high, we are merely observing drawbacks of the current out-of-the-box approaches for this endeavor.  In fact, we take a moment to catalogue some potentially powerful observations made while tinkering with Example~\ref{ex0}.

%

\begin{remark} {\bf (Interesting observations of SVD and highly-cyclic structure)}
We list a few attractive aspects of singular value embeddings for detecting highly-cyclic structure.
\bit
\item There is high potential in using regions of $A_{bp}$ where average path lengths are increased over those of $A$ as a indicator of highly-directed structure.
\item Coordinates of vertices in classical community structure are mapped to similar locations in the two embeddings associated with the left singular vector and  the right singular vector. 
\item Coordinates of vertices involved in cyclic community structure are mapped to dissimilar locations in the two embeddings associated with the left singular vector and the right singular vector.   In Figure~\ref{fig:ex0b} their locations are rotated one-third of the circle.   It is likely these types of embedding properties can be leveraged for detecting various classes of community structure or highlighting regions of highly-directed flow.
\eit
\end{remark}

The next sections discuss how a planar embedding from a single complex-vauled eigenpair can overcome some of the challenges associated with the SVD approach.  The approaches developed in this paper are applied to the network from Example \ref{ex0} in Section \ref{sec:experiments1}.



\section{The 3-cyclic case}
\label{sec:global_structure}
The simplest example of highly 3-cyclic structure is that of a purely 3-cyclic graph, $G=(V,E)$.  Here, the vertex set can be partitioned into three non-overlapping sets  $V = V_0 \cup V_1 \cup V_2$ and $E = \{(i,j)| i \in V_l, j \in V_{(l+1) \mod 3} \}$.  That is, edges only flow in a directed 3-cycle around supernodes $V_0, V_1,$ and $V_2$.  In this section, we examine the eigenvalues and eigenvectors of graphs with this purely 3-cyclic structure.

\subsection{Stateful graphs and 3-cyclic structure}

Given $G(V, E)$, one way to identify purely (or highly) 3-cyclic structure in the graph is to build a stateful graph, $\cG_3(\cV, \cE)$.  This is done in the following way: for each $i \in V$, make three copies, $i$, $i+n$, and $i+2n \in \cV$, called the {\em red}, {\em green}, and {\em blue} versions of $i$, respectively.    Then, for each $(j, i) \in E$, we also have three copies, $(j+2n, i), (j, i+n)$, and $(j+n, i+2n) \in \cE_3$.   See Figure~\ref{fig:cyc3} for a visual example.

\begin{figure}[h]
\centering
\includegraphics[draft = false, width = 3in]{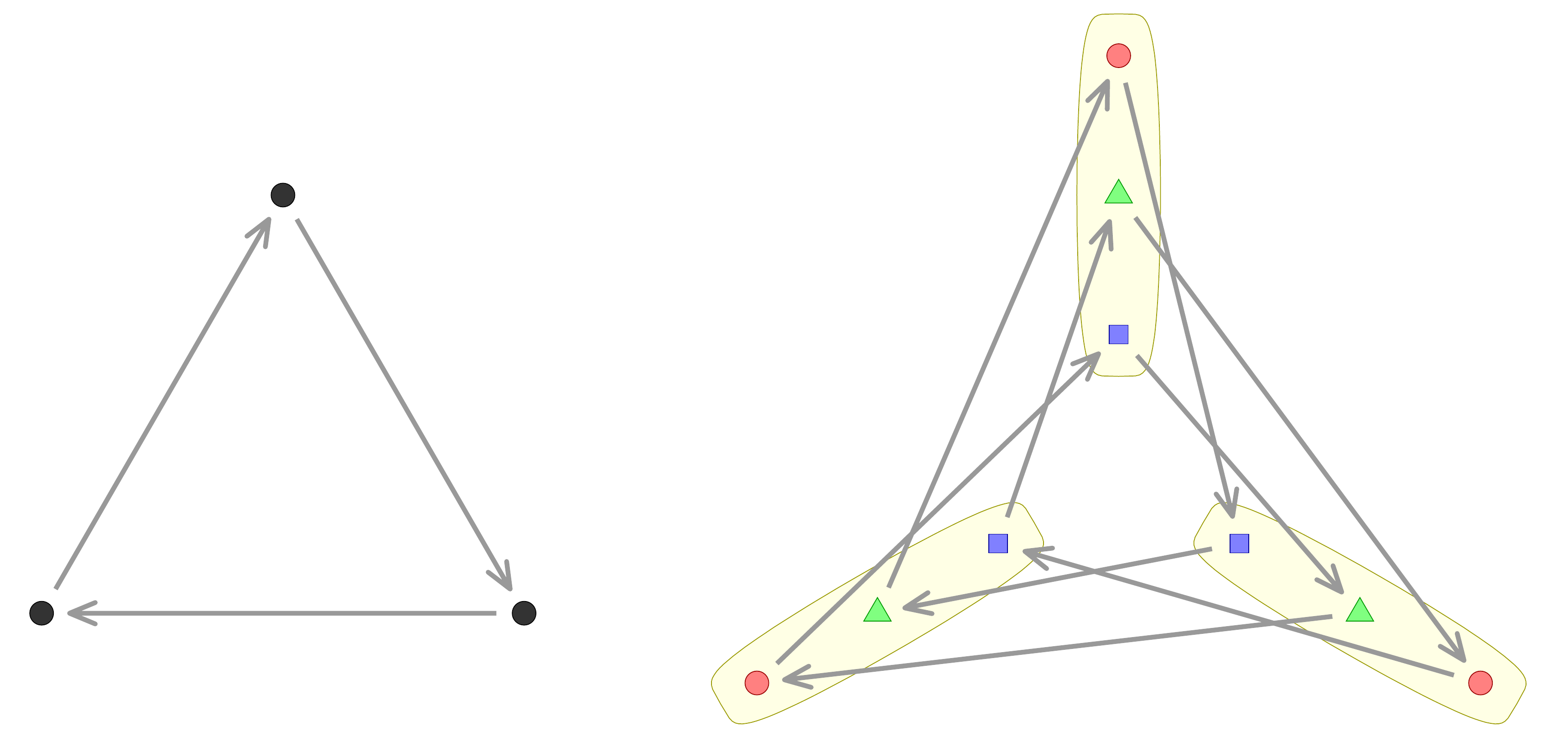} 
\includegraphics[draft = false, width = 3in]{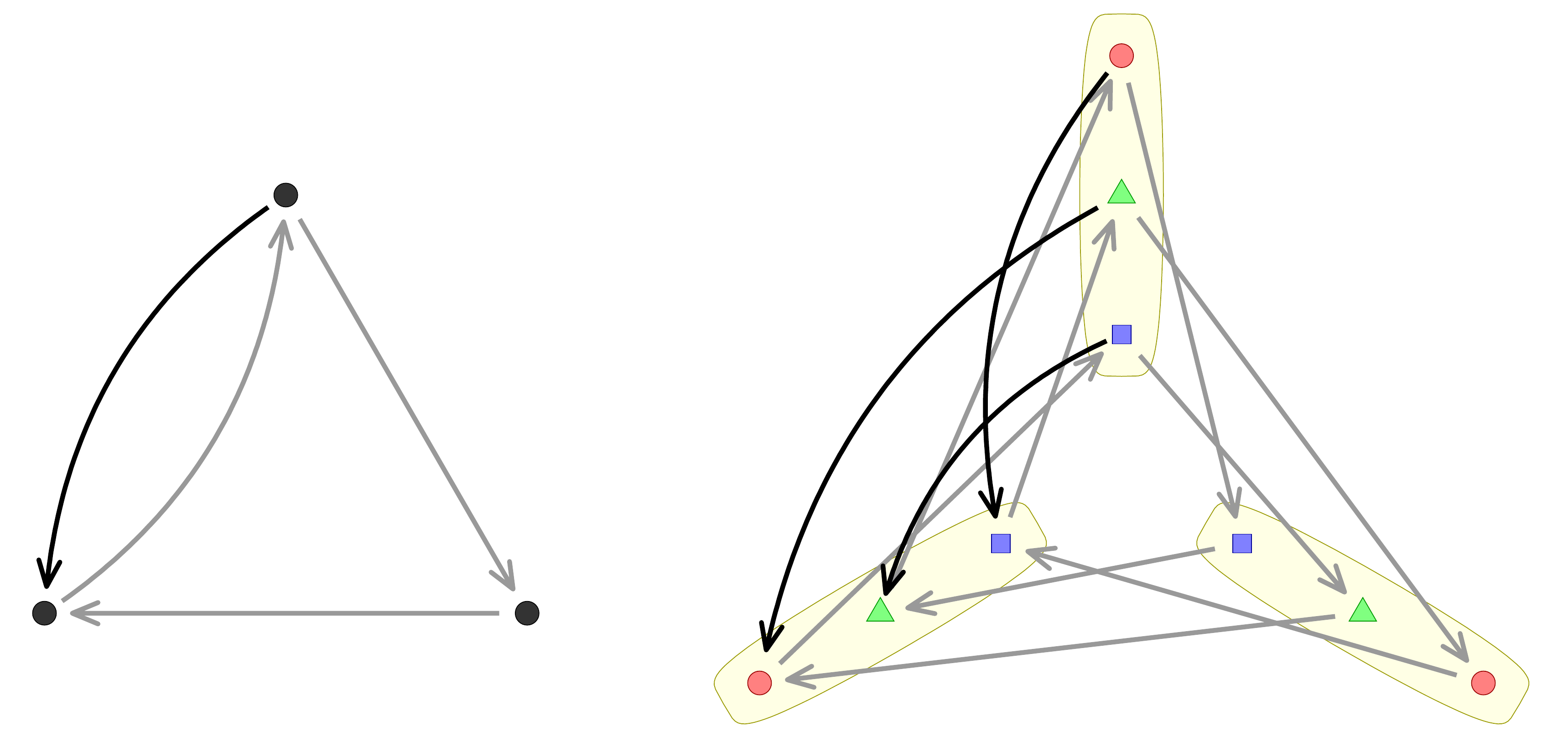} 
\caption{(left) a directed 3-cycle and its associated (center-left)  3-color state space; (center-right) a 3-cycle with a single reciprocal edge, and (right) a 3-cycle with a single reciprocal edge, 3-color state space.   }
\label{fig:cyc3}
\end{figure} 

The stateful graph provides a new topology where highly 3-cyclic structure in the original graph $G$ becomes evident.   Consider a portion of $G$ that is highly 3-cyclic:   there exist many short paths from a vertex $i$ around a 3-cycle back to itself in $G$.   However, in the stateful graph, we may not have any short paths from the red $i$ to  other colors of the same vertex, $i+n$, $i+2n$.   In fact, it is necessary that a path leave the 3-cyclic structure, either encountering a reciprocal edge or a cycle larger than 3, for $i$ to reach $i+n$ or $i+2n$.   Thus, if the {\em average self-distances} between $i$ and $i+n$ are relatively large, then $i$ is part of highly 3-cyclic structure.   Again, see Figure~\ref{fig:cyc3}.      


This concept is general, and a wide array of graph computations could be employed on $\cG_3$.    In this paper, we focus on spectral methods and  demonstrate some attractive properties of eigenvector techniques for identifying highly cyclic structure in both $G$ and $\cG_3$.

Let $A$ be the adjacency matrix associated with a digraph $G$.   Let $D$ be the out-degree matrix, $D = $diag$(\bfo^tA)$, so $B = D^{-1}A$ is the stochastic transition matrix associated with $G$.  Then, the stochastic transition matrix associated with $\cG_3$ is
$$
B_3 = 
\left[ 
\barr{ccc}
0 & 1 & 0 \\
0 & 0 & 1 \\
1 & 0 & 0 
\earr
\right]
 \otimes
 B=
\left[ 
\barr{ccc}
0 & B & 0 \\
0 & 0 & B \\
B & 0 & 0 
\earr
\right],
$$
where $\otimes$ denotes the Kronecker matrix product.  Because each vertex has an identical statespace and the edges are wired in a uniform way, the eigendecomposition of $B_3$ has a simple relationship with that of $B$, specifically $\sigma(B_3)$ is three rotated copies of $\sigma(B)$.

\begin{theorem}\
\label{thm:B3eig}
Let $(\bv, \lambda)$ be a right  eigenpair for $B$, $B \bv  = \lambda \bv$.     For $p=0,1,2$, we have
\beq
\label{eqn:B3eig}
B_3 
\left[ 
\barr{c}
\theta_{0,3} \bv \\
\theta_{p,3} \bv \\
\theta_{2p,3} \bv
\earr
\right] 
=
(\theta_{p,3} \lambda) 
\left[ 
\barr{c}
\theta_{0,3} \bv \\
\theta_{p,3} \bv \\
\theta_{2p,3} \bv
\earr
\right].
\eeq
\end{theorem}
\begin{proof}
This is a consequence of a general result regarding Kronecker products of matrices and their eigenpairs.   We verify it in this case for the sake of completeness.   For $r = 0,1,2$, the $(r+1)$-th row block of Equation~(\ref{eqn:B3eig}) is verified by 
$$
B \theta_{rp+p,3} \bv   =   \theta_{rp+p,3} B \bv \,=\, \theta_{rp+p,3} \lambda \bv \,=\,  (\theta_{p,3} \lambda) (\theta_{rp,p} \bv ).
$$
\end{proof}


\subsection{Algebraic 3-periodicity}

In terms of a random walk on the stateful graph, if there is an eigenvalue in $\sigma(B)$ such that $\lambda \approx \theta_{1,3}$, then the second eigenvalue in $\sigma(B_3)$ has a value close to, but not equal to, $1$.   The associated eigenvector is a slowly mixing mode with respect to $\cG_3$.    This happens when highly 3$k$-cyclic structure is present.  Thus, the eigenvalues of $B$ and $B_3$ can be used to identify cyclic structure in directed graphs.   We prove the $\lambda = \theta_{1,3}$ case below.

\begin{theorem} Let the graph associated with $B$ be strongly connected. Then, $\theta_{1,3} \in \sigma(B)$ if and only if the graph associated with $B$ is $3k$-cyclic, for some positive integer $k$.
\label{thm:3-cyclic_evalues}
\end{theorem}
\begin{proof}
If $B$ is $3k$-cyclic for some positive integer $k$, then $B$ can be written in the following block form: $$B = \left[ 
\barr{ccc}
0 & 0 & C_2 \\
C_0 & 0 & 0 \\
0 & C_1 & 0 
\earr
\right] $$
where $C_0, C_1,$ and $C_2$ are row stochastic.  Now, $$B^3 = \left[ 
\barr{ccc}
C_2C_1C_0 & 0 & 0 \\
0 & C_0C_2C_1 & 0 \\
0 & 0 & C_1C_0C_2 
\earr
\right].$$ 
is also row stochastic and has eigenvalue $\lambda = 1$ of multiplicity of at least three.  Now, there are at least three eigenvalues of $B$ of the form $\lambda  = \sqrt[3]{1}$.  However, since $B$ is irreducible, by the Perron-Frobenius theorem, $\lambda_1 = 1$ is an eigenvalue of $B$ with multiplicity 1.  The remaining eigenvalues of the form $\lambda  = \sqrt[3]{1}$ must come in pairs of $\theta_{1,3}, \overline{\theta_{1,3}}=\theta_{2,3}$.  Thus, $\theta_{1,3}$ is an eigenvalue of $B$.


Next, if $\theta_{1,3} \in \sigma(B)$, then $\theta_{2, 3} \in \sigma(B)$, due to fact that complex eigenvalues of real matrices come in conjugate pairs and that $\overline{\theta_{1,3}} = \theta_{2, 3}$.  Additionally, since $B$ is irreducible, the Perron-Frobenius theorem states that the period of $B$ is given by $p$, where $p$ is the the number of eigenvalues $\lambda$ with $|\lambda| = \rho(B) = 1$ and each of these eigenvalues is a $p$th root of unity.  As $\theta_{1,3}$ is only a $p$th root of unity when $p = 3k$ for some integer $k>0$, the period of $B$ is $p=3k$.  Since $p>1$, the Perron-Frobenius theorem also states that there exists a permutation matrix $P$ such that   
$$
PBP^{-1}= \left[ 
\barr{ccccc}
0 & 0 & 0 & \cdots & C_{p-1} \\
C_0 & 0 & \cdots & 0 & 0 \\
0 & C_1 & 0 & \cdots & 0 \\
 &  & \ddots &  &  \\
0 & 0 & \cdots & C_{p-2} & 0
\earr
\right]. 
$$
Thus, $B$ is $3k$-cyclic.
\end{proof}

\begin{remark}{\sc ($3k$-Cyclic Structure)}
Although we are concerned with finding $3$-cyclic structure, all of the methodologies presented here discuss finding regions of $3k$-cyclic structure for some positive integer $k$.  This is due to the fact that for $k>1$, any $3k$-cyclic structure with vertex groups $V_0, V_1, V_2, \ldots V_{3k-1}$ can also be viewed as $3$-cyclic with vertex groups $V_0 \cup V_3 \cup \ldots V_{3k-3}$, $V_1 \cup V_4 \cup \ldots \cup V_{3k-2}$, and $V_3 \cup V_5 \cup \ldots \cup V_{3k-1} $ and the row stochastic adjacency matrix will also have $\theta_{1,3}$ as an eigenvalue.
\end{remark}

\subsection{Spectral Coordinates}
\label{subsec:spectral}

Although it is nice to identify the existence of highly cyclic structure in a graph, it is often more important to classify the nodes of a network based on their participation in this structure.  This can be done using the eigenvectors associated with $\lambda = 1, \theta_{1,3},$ and $\theta_{2,3}$.  Let $B$ be strongly connected and 3-cyclic and let $V_0, V_1, V_2 \subset V$ be the three sets of nodes which make up the nontrivial strongly connected components of the graph of $B^3$.   By Perron-Frobenius theorem, for each of $V_0, V_1, V_2$ there exists both a left and a right real-valued eigenvector of $B^3$ associated with $\lambda=1$ that is positive on the nodes in the component and zero outside.   For the right eigenvectors, let the positive part on each $V_i$ be labeled $\bv_i$.    Then (potentially after node relabeling), the eigenspace of $B^3$ associated with $\lambda = 1$ is spanned by 

$$
\left[
\barr{ccc}
\bv_0 & 0 & 0 \\
0 & \bv_1 & 0 \\
0 & 0 & \bv_2 
\earr
\right].
$$   
Since $B^3$ is row stochastic and these are eigenvectors associated with the eigenvalue $\lambda=1$, each $\bv_i$ for $i =0,1,2$ must be a constant vector.
 
The the right eigenspaces of $B$ associated with $\lambda =1, \, \theta_{1,3},$ and  $\theta_{2,3}$ are also spanned by this basis.   We can rotate the basis of this span, using the methodology from Theorem \ref{thm:B3eig}, to form an equivalent basis:
$$
\left[
\barr{rrr}
\alpha_0 \bv_0 & \alpha_0 \bv_0 & \alpha_0 \bv_0 \\
\alpha_1 \bv_1 & \theta_{1,3} \alpha_1 \bv_1 & \theta_{2,3} \alpha_1 \bv_1 \\
\alpha_2 \bv_2 & \theta_{2,3} \alpha_2 \bv_2 & \theta_{1,3} \alpha_2 \bv_2 
\earr
\right] \\
$$   
where $\alpha_i > 0$ for $i=0,1,2$ are positive scalers.

Similarly, the left eigenspace of $B^3$ associated with $\lambda =1$ is spanned by
$$
\left[
\barr{ccc}
\bu_0 & 0 & 0 \\
0 & \bu_1 & 0 \\
0 & 0 & \bu_2 
\earr
\right]
$$ 
where $\bu_i \geq 0$although, here, the $\bu_i$'s are not necessarily constant.  The left eigenspace of $B$ associated with $\lambda =1, \, \theta_{1,3},$ and  $\theta_{2,3}$ is also spanned by this basis.  As in the case of the right eigenspace, this basis can be rotated as follows
$$
\left[
\barr{rrr}
\beta_0 \bu_0 & \beta_0 \bu_0 & \beta_0 \bu_0 \\
\beta_1 \bu_1 & \theta_{2,3} \beta_1 \bu_1 & \theta_{1,3} \beta_1 \bu_1 \\
\beta_2 \bu_2 & \theta_{1,3} \beta_2 \bu_2 & \beta_{2,3} \beta_2 \bu_2 
\earr
\right] \\
$$   
where $\beta_i >0$ for $i=0,1,2$ are positive scalers.

\begin{theorem}
Let $G$ be a strongly connected $3k$-cyclic graph with stochastic transition matrix $B = D^{-1}A$ that can be written in block form $$B = \left[ 
\barr{ccc}
0 & 0 & C_2 \\
C_0 & 0 & 0 \\
0 & C_1 & 0 
\earr
\right]. $$ Further, let $\bx$ be a right eigenvector of B associated with eigenvalue $\lambda = \theta_{1,3}$, where $$\bx^* = [(\alpha_0 \bv_0)^*, (\theta_{1,3} \alpha_1 \bv_1)^*, (\theta_{2,3} \alpha_2 \bv_2)^*].$$  Then, the entries of $\bx$ cluster the nodes in the network according to their membership in $C_0$, $C_1$, or $C_2$.
\label{thm:sorting}
\end{theorem}

\begin{proof}
For any $\lambda \in \sigma(B)$, $\lambda$ can be decomposed as $\lambda = \rho \exp(\iota 2 \pi \phi)$, with $\rho \geq 0$ and $\phi \in [0,1)$.  Let $\bx$ be the normalized right eigenvector associated with $\lambda.$  The $i$th entry in $\bx$, $x_i$, can be decomposed similarly, with $x_i = p_i \exp(\iota 2 \pi t_i)$.  Then, for $x_i \neq 0$ and $\cN_i$ the set of nodes in the out-neighborhood of node $i$ (that is, $(i,j) \in E(G)$), 

$$
\lambda  =  \frac{1}{d_i}\sum_{j \in \cN_i} \frac{x_j}{x_i} 
$$
can be rewritten as
\beq
\label{eqn:rhoandt}
\rho \exp(\iota 2 \pi \phi)  =   \frac{1}{d_i}\sum_{j \in \cN_i} \frac{p_j}{p_i} \exp(\iota 2 \pi (t_j - t_i)).
\eeq
Now, by applying absolute values, the triangle inequality gives:
\beq
\label{eqn:rho}
\rho \leq \frac{1}{d_i}\sum_{j \in \cN_i} \frac{p_j}{p_i} \leq \max_{j \in \cN_i} \frac{p_j}{p_i}.
\eeq    

In the case of $\lambda = \theta_{1,3}$ and $B\bx = \lambda \bx$, $\rho = 1$ and $p_j$ is constant for all $j \in V(G)$. To see that this is the case, let $i$ be a vertex such that $p_i \geq p_j$ for all vertices $j \neq i$.  Then, by Equation \ref{eqn:rho}, 
$$
1 \leq \frac{1}{d_i}\sum_{j \in \cN_i} \frac{p_j}{p_i} \leq \max_{j \in \cN_i} \frac{p_j}{p_i}.
$$  
If $p_j < p_i$ for any $j \in \cN_i$, the first inequality would no longer hold.  As $G$ is a strongly connected graph, the equality can be extended among all nodes in the network.

Now,  for any (directed) edge $(j,i) \in E(G)$, we have $\exp(\iota 2 \pi (t_j - t_i)) = \exp(\iota 2 \pi \phi)$.  This follows from plugging in $\rho = 1$ and $p_i = p_j$ for all $i,j \in V(G)$ into Equation \ref{eqn:rhoandt}:  
$$
\exp(\iota 2 \pi \phi)  =  \frac{1}{d_i}\sum_{j \in \cN_i} \exp(\iota 2 \pi (t_j - t_i)).
$$
This forces all complex numbers in the summation to have the same argument, $\iota 2 \pi \phi$, which is equal to $\frac{\iota 2 \pi}{3}$ in the case of $\lambda = \theta_{1,3}$.  Finally, this shows that as one moves across edge $(i, j) \in G(E)$ the phase shift of an eigenvector $\bx$ associated with $\lambda = \theta_{1,3}$ is exactly $\phi = \frac{1}{3}$.  

All together, this means that each entry in $\bx$ can be mapped to a vector in the $xy$-plane with a magnitude of at most 1 and an angle of $\frac{2\pi}{3}, \frac{4\pi}{3}$ or $2\pi$.  The formula for this mapping can be found in Lemma \ref{lemma:mapping}. This clusters the nodes of $G$ by their entries of $\bx$ into three groups, corresponding to membership in $C_0, C_1$, or $C_2$. \\
\end{proof}

\begin{lemma}
Let $(\lambda, \bv)$ and $(\overline{\lambda}, \overline{\bv})$ be eigenpairs of $B$ such that $\mbox{Im} \, \lambda \neq 0$.   Let $\bv = \br + \iota \, \bc$.   Consider the two-dimensional spectral coordinates $(v_i, \overline{v_i})$.   There exists a 2d complex orthogonal rotation that places these coordinates in $\mathbb{R}^2$: 
$$
\ba = \frac{1}{\sqrt{2}} \left( \bv + \overline{\bv} \right) = \sqrt{2} \br
\qq{and}
\bb =  \frac{1}{\sqrt{2} \iota} \left( \bv - \overline{\bv} \right) = \sqrt{2} \bc. \\
$$
\label{lemma:mapping}
\end{lemma}

The results from Theorems \ref{thm:3-cyclic_evalues} and \ref{thm:sorting} can be visualized on a graph containing purely 3-cyclic structure.  We build an example of such a graph using a stochastic block model with three specific groups of nodes, $V_0, V_1,$ and $V_2$, each of size 45.  Here, the probability of an edge from $V_0$ to $V_1$, from $V_1$ to $V_2$, or from $V_2$ to $V_0$ is given by $\rho = 0.8$ and the probability of any other edge is 0.  The adjacency matrix of one instance of the resulting graph can be seen on the left of Figure \ref{fig:ex0} and the associated spectrum of the row-stochastic adjacency matrix, $B$, can be seen on the right.  As expected, $\theta_{1,3}$ and $\theta_{2,3}$ are both eigenvalues of $B$.

\begin{figure}[h]
\centering
\includegraphics[width = .515\textwidth]{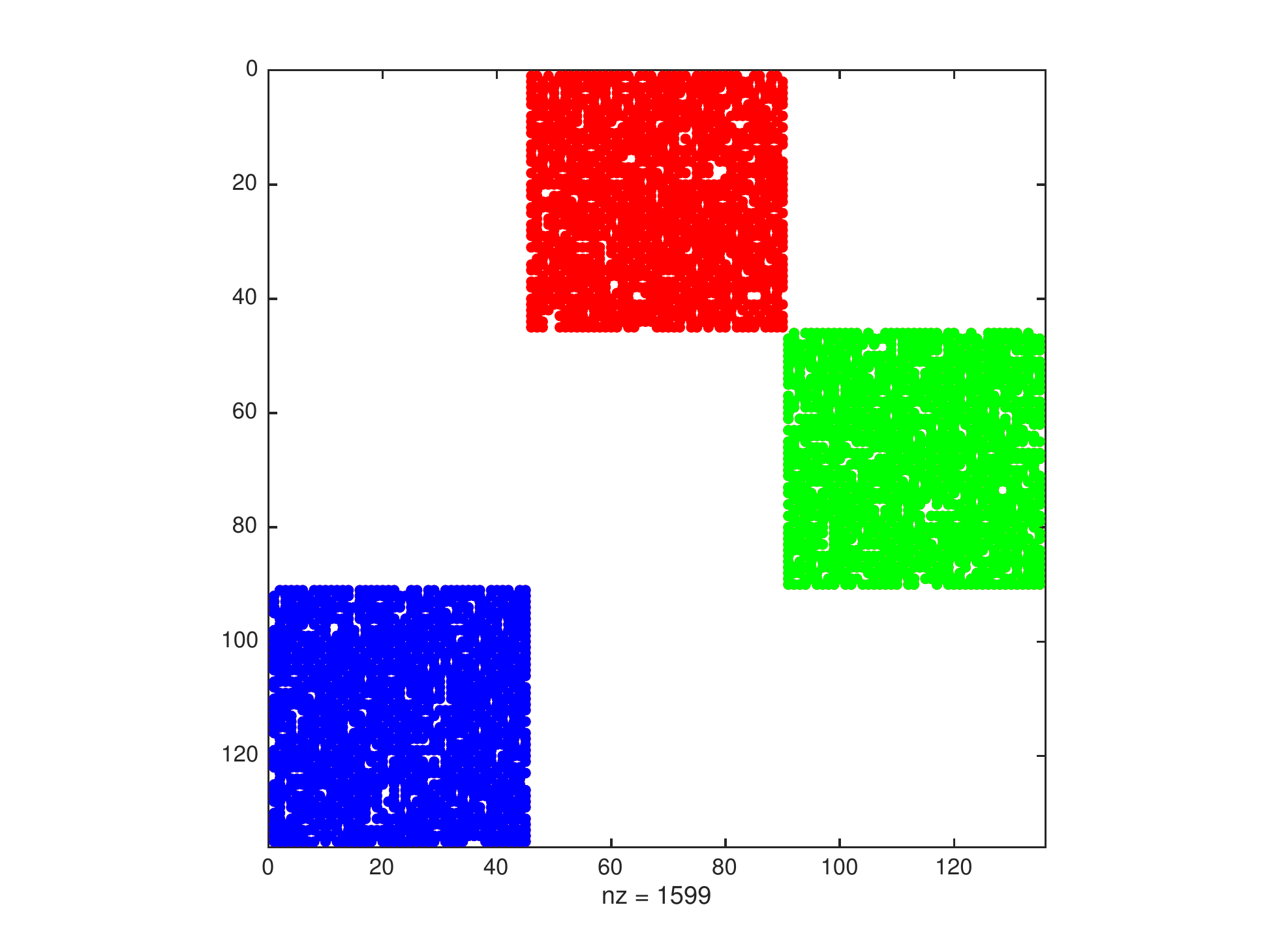} 
\hspace{0.15in}
\includegraphics[width = .44\textwidth]{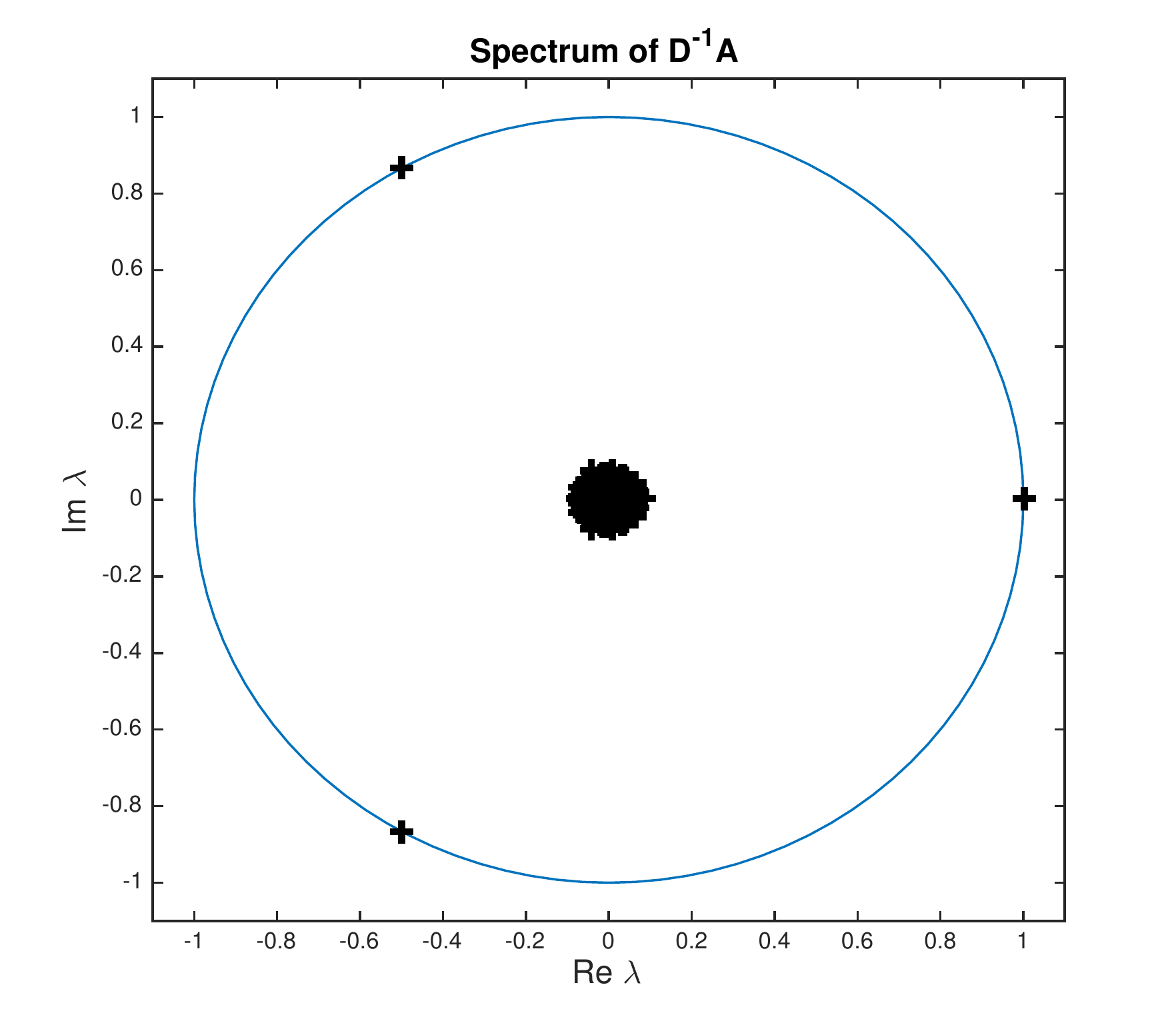} 
\caption{The adjacency matrix (left) and spectrum (right) of  a purely 3-cyclic graph.}
\label{fig:ex0}
\end{figure} 

The embedding of the nodes of the network into $\mathbb{R}^2$ using the left and right eigenvectors associated with $\lambda=\theta_{1,3}$ are displayed in Figure \ref{fig:exscs}.  In the embedding formed using the right eigenvector, the (red) nodes in $V_0$ are mapped onto a single point at an angle of $\frac{4\pi}{3}$, the (green) nodes from $V_1$ are mapped onto a single point at an angle of $\frac{2\pi}{3}$, and the (blue) nodes in $V_2$ are mapped onto a single pout at an angle of $2\pi$.  When the left eigenvector is used to embed the nodes, the three groups are also identified.  In this case, each node in $V_0$ is mapped to an angle of $\frac{4\pi}{3}$, but with a range of magnitudes.  Similarly, the nodes from $V_1$ are all mapped to an angle of $\frac{2\pi}{3}$ and the nodes from $V_2$ are all mapped to an angle of $2\pi$.

\begin{figure}[h]
\centering
\includegraphics[width = 0.48\textwidth]{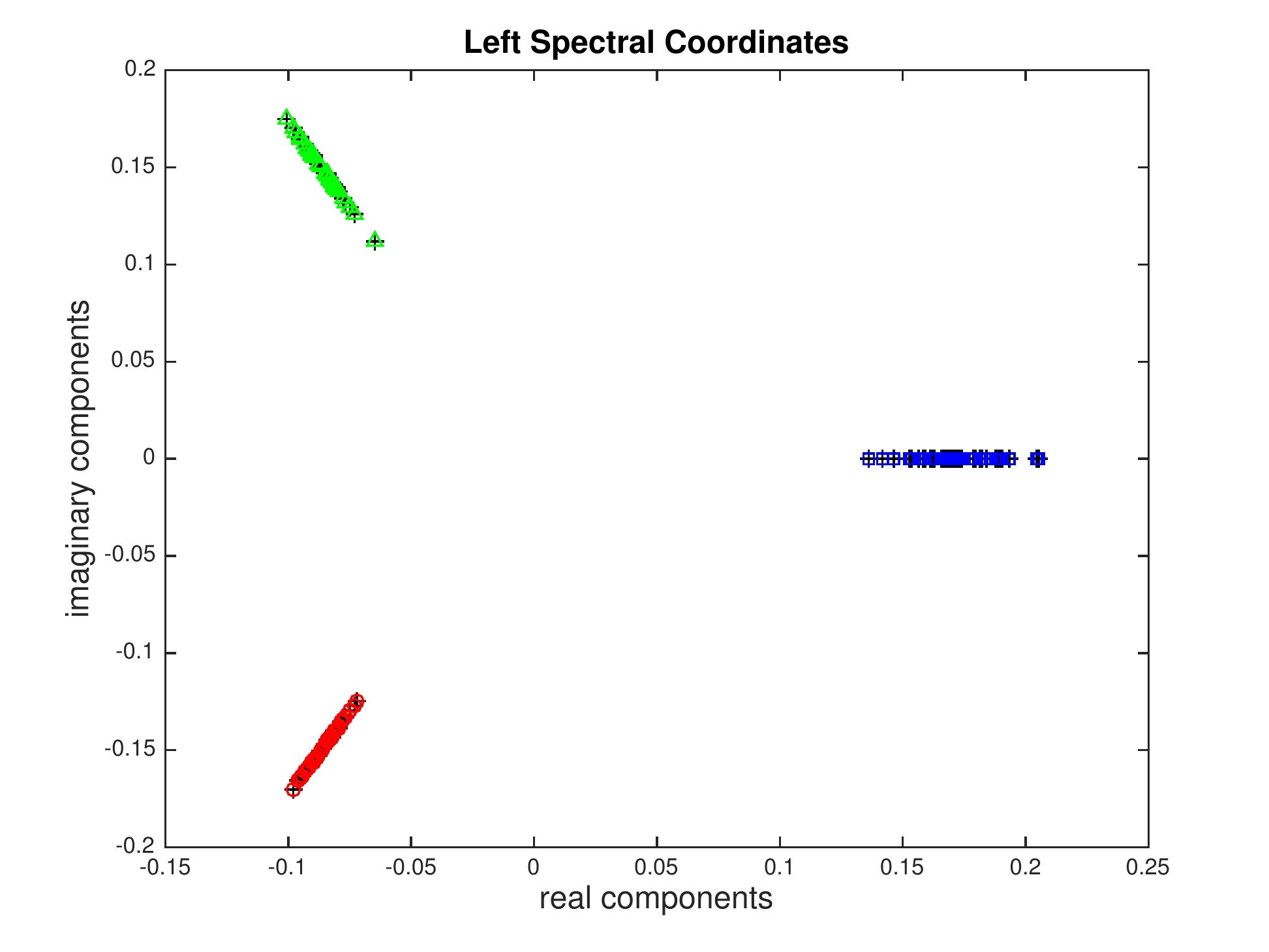} 
\hspace{0.15in}
\includegraphics[width = .48\textwidth]{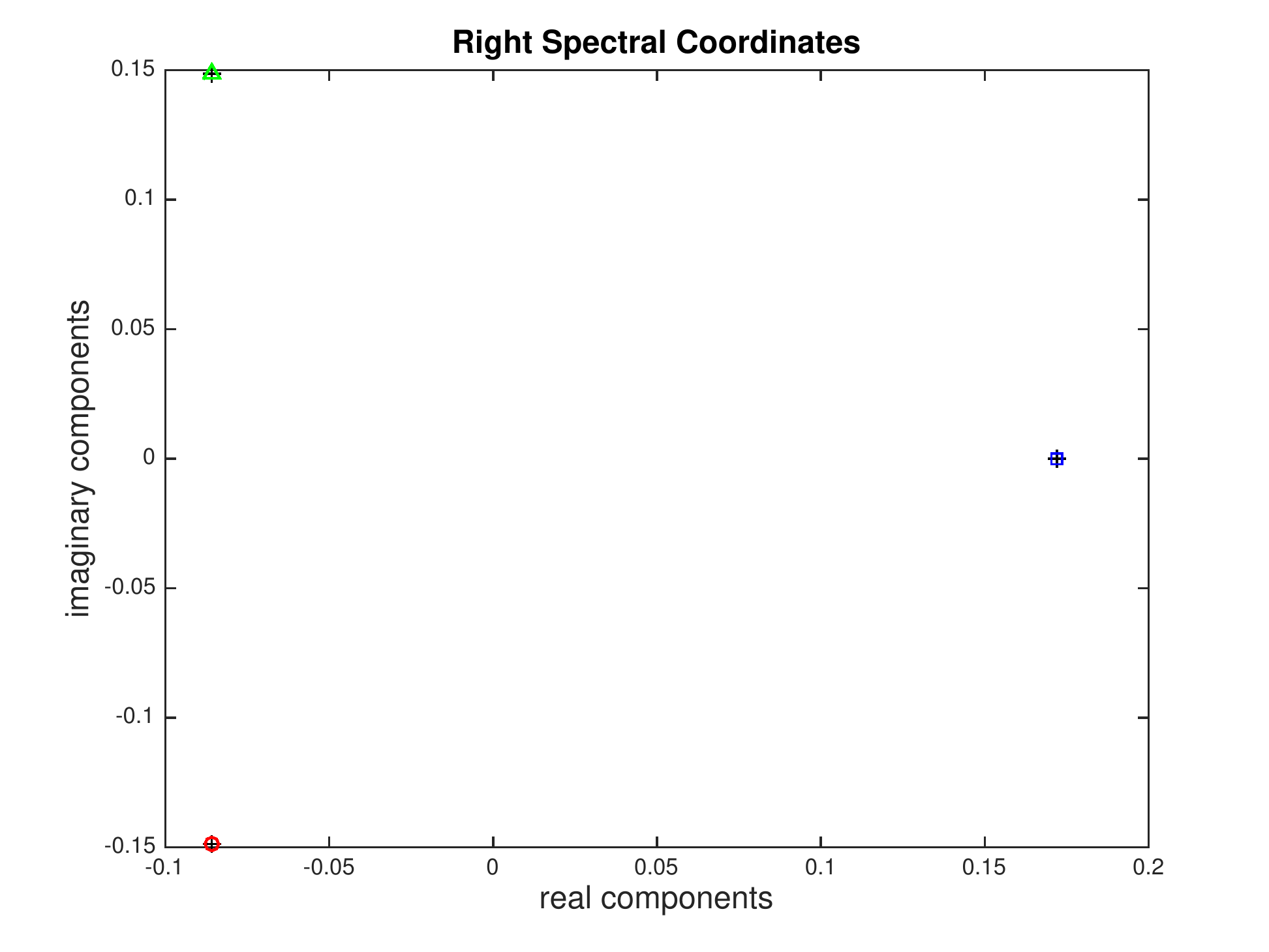} 
\caption{Coordinate embeddings using the left and right eigenvectors associated with $\lambda = \theta_{1,3}$.  In the embedding associated with the right eigenvector (right), the nodes in the same group are embedded to the exact same point while in the embedding using the left eigenvector (left), they are embedded at the same angle, but with different magnitudes.}
\label{fig:exscs}
\end{figure} 

Theorem \ref{thm:sorting} and Lemma \ref{lemma:mapping} show how sorting by angle completely reveals the sets $V_0, V_1$, and $V_2$.   Of course, for a truly 3-cyclic graph, this is not the most efficient manner to classify the nodes of $G$. A breadth-first-search approach accurately labels these sets and is much faster.   In order for the linear algebraic approach to be useful, it needs to be extended to the case where $G$ is a highly, but not purely, 3-cyclic graph or has regions of highly 3-cyclic structure.  Some results concerning this fuzzy 3-cyclic case can be found in Section \ref{sec:fuzzy}.

\begin{remark}
Similarly, the nodes of $G$ can be classified into the three groups using the stateful graph $\cG$ and its stochastic transition matrix $B_3$.   However,  $B_3$ is three times the size of $B$.  Forming $B_3$ and calculating several eigenpairs with eigenvalues close to 1 is not necessary for computing desired spectral coordinates.
Instead, one only needs to compute members of eigenspaces of $B$ with eigenvalues near $\theta_{1,3}$ and use their real and imaginary parts to organize vertices.
\end{remark}

\section{The fuzzy 3-cyclic case}
\label{sec:fuzzy}

The eigenvector approach to classifying nodes into clusters based on 3-cyclic structure in a graph is most useful when it does so in a graph that is not purely 3-cyclic, but instead has a dominant 3-cyclic structure (or region of 3-cyclic substructure) plus added noise.  

\begin{lemma}
Given a matrix $B$ associated with a purely 3-cyclic graph as described in Theorem \ref{thm:3-cyclic_evalues} with vertex set $V = V_0 \cup V_1 \cup V_2$, let $\bx$ and $\by$ be the normalized right and left eigenvectors of B associated with $\theta_{1,3}$.  Then, $| \by^* \bx |$ is bounded below by $| \by^* \bx | \geq \frac{1}{\sqrt[4]{|V|}}$.
\label{lemma:y*x}
\end{lemma}

\begin{proof}
Since $B$ is a purely 3-cyclic graph, as seen in Section \ref{subsec:spectral}, $\bx$ and $\by$ have the form 
$$
\bx = \left[
\barr{r}
\alpha_0 \bv_0 \\
\theta_{1,3} \alpha_1 \bv_1 \\
\theta_{2,3} \alpha_2 \bv_2 
\earr
\right] 
{\rm \, and \,\,}
\by = \left[
\barr{r}
\beta_0 \bu_0 \\
\theta_{2,3} \beta_1 \bu_1 \\
\theta_{1,3} \beta_2 \bu_2 
\earr
\right] \\
$$
where $\alpha_0^2 + \alpha_1^2 + \alpha_2^2 = 1$, $\beta_0^2 + \beta_1^2 + \beta_2^2 = 1$, and $\|\bv_i\|_2 = \|\bu_i\|_2 =1$ for $i=0,1,2$. Additionally, $\bv_i$ is constant for $i=0,1,2$.  Now,
$$\| \by^*\bx\|_2^2 = | \alpha_0 \beta_0 \bu_0^* \bv_0 + \alpha_1 \beta_1 \bu_1^* \bv_1 + \alpha_2 \beta_2 \bu_2^* \bv_2 |
$$ 
\beq
\geq \alpha_0 \beta_0 \frac{1}{\sqrt{|V_0|}} \|\bu_0\|_2 + \alpha_1 \beta_1 \frac{1}{\sqrt{|V_1|}} \|\bu_1\|_2 + \alpha_2 \beta_2 \frac{1}{\sqrt{|V_2|}} \|\bu_2\|_2 
\label{eqn:2norm}
\eeq
$$
= \alpha_0 \beta_0 \frac{1}{\sqrt{|V_0|}} + \alpha_1 \beta_1 \frac{1}{\sqrt{|V_1|}} + \alpha_2 \beta_2 \frac{1}{\sqrt{|V_2|}}.
$$
Due to the fact that $B\bw = \bw$ for
$$\bw = \left[
\barr{r}
\alpha_0 \bv_0 \\
\alpha_1 \bv_1 \\
\alpha_2 \bv_2 
\earr
\right],  $$
it follows that $\alpha_i = \sqrt{\frac{|V_i|}{|V|}}$ for $i=0,1,2$.  Combined with (\ref{eqn:2norm}) above and the fact that $\beta_i > 0$ for $i=1,2,3$, we get
$$
\| \by^*\bx\|_2^2 \geq \frac{1}{\sqrt{|V|}}(\beta_0 + \beta_1 + \beta_2) \geq \frac{1}{\sqrt{|V|}}.
$$
Thus, $| \by^*\bx| \geq \frac{1}{\sqrt[4]{|V|}}$.
\end{proof}

\begin{theorem}
Let $G$ be a graph with two strongly connected components and an adjacency matrix that can be written in block form 
$$A = \left[ \barr{cccc}
A_o & & & \\
& 0 & 0 & C_2 \\
& C_0 & 0 & 0 \\
& 0 & C_1 & 0 
\earr
\right]$$
where $A_o$ is the adjacency matrix of $G_o$, a strongly connected graph that is not 3k-cyclic for any integer $k$, i.e. $\lambda = \theta_{1,3}$ is a simple eigenvalue of $A$.  Let $B = D^{-1}A$ be the stochastic transition matrix associated with $G$. Let $\hat{G}$ be $G$ with noise added in the zero blocks of $G$.  The stochastic row transition matrix of $\hat{G}$ can be written as $\hat{B} = \hat{D}^{-1}\hat{A} = D^{-1}A  +M$.  Then, there exists $\lambda_o \in \sigma(\hat{B})$ such that
$$ 
\left| \lambda_o - \theta_{1,3}\right| < 2\sqrt[4]{C}\left( \max_i  \frac{\hat{d}_i - d_i}{\hat{d}_i}\right)  + \bigO{\left(\max_i  \frac{2(\hat{d}_i - d_i)}{\hat{d}_i} \right)^2}
$$ 
 where $\hat{d}_i$ is the out-degree of node $i$ in $\hat{G}$, $d_i$ is the out-degree of node $i$ in $G$, and $C = |V_0| + |V_1| +|V_2|$ is the number of nodes in the 3-cyclic region of the network. \\
\label{thm:fuzzy_sorting}
\end{theorem}

\begin{proof}
From \cite[p.183]{StSu90}, we have that there exists $\lambda_o \in \sigma(A_o)$ such that $$\lambda_o = \theta_{1,3} + \frac{\| \by^* M \bx \|_2}{\| \by^* \bx \|_2} + \bigO{\| M \|^2}.$$  Thus, 
$$
| \lambda_o - \theta_{1,3}| = \left| \frac{\| \by^* M \bx \|_2}{\| \by^* \bx \|_2} + \bigO{\| M \|_2^2} \right|  \leq \frac{\|\by^*\|_2 \|M\|_2 \|\bx\|_2}{\| \by^*\bx \|_2} + \bigO{\| M \|_2^2} 
$$ 
$$ 
= \frac{\|M\|_2}{\| \by^*\bx \|_2} + \bigO{\| M \|_2^2} = \frac{\lambda_{max}(M)}{\| \by^*\bx \|_2} + \bigO{\lambda_{max}(M)^2}
$$
by \cite[p. 282]{Me00}.

Now, by Gershgorin's circle theorem, $\lambda_{max}(M) \leq \max \sum_{j = 1, j \neq i}^n |M_{ij}|$.  As $M = \hat{D}^{-1}\hat{A} - D^{-1}A$, the entries of $M$ are given by:
$$M_{ij}=\left\{\begin{array}{ll}
\frac{1}{\hat{d}_i},& \textnormal{ if } (i,j) \in E(\hat{G})/E(G),\\
\frac{1}{\hat{d}_i} - \frac{1}{d_i} & \textnormal{ if } (i,j) \in E(G),\\
0, & \textnormal{ else }
\end{array}\right .
$$
where $\hat{d}_i$ is the degree of node $i$ in $\hat{G}$.  Thus, for a fixed $i$, $\sum_{j = 1}^n |M_{ij}| =d_i|\frac{1}{\hat{d}_i} - \frac{1}{d_i}| + (\hat{d}_i - d_i)d_i = \frac{2(\hat{d}_i - d_i)}{\hat{d}_i}$.  Combined with the results from Lemma \ref{lemma:y*x}, the theorem follows.
\end{proof}

The bounds presented in Lemma \ref{lemma:y*x} and Theorem \ref{thm:fuzzy_sorting} work well when the 3-cyclic region in the larger network is relatively small and well-separated, that is when both $\sqrt[4]{C}$ and  $\left(\max_i  \frac{\hat{d}_i - d_i}{\hat{d}_i}\right)$ are small.  As the 3-cyclic region gets larger and/or more connected to the rest of the network, the bounds presented in the above theorem increase above $\rho(B) = 1$ and lose usefulness.  However, experimental results suggest that $|\lambda_o - \theta_{1,3}|$ is often small, even in networks where the above bounds are large.   

The magnitude of the $ \bigO{\left(\max_i  \frac{2(\hat{d}_i - d_i)}{\hat{d}_i} \right)^2}$ term is governed by how close to simple $\theta_{1,3}$ is as an eigenvalue of $A$.  If there are several highly cyclic structures in $G_o$, leading to one or more eigenvalues of $A$ close to $\theta_{1,3}$, the magnitude of the higher order terms in the bound will increase.  A detailed discussion of the effects of this on the approximation of $\lambda_o$ is outside the scope of this paper, but in various experiments it seems small.  A more detailed discussion on the approximation of the second order terms for general matrix perturbations can be found in \cite[Ch. 5]{StSu90} and it may be possible to tighten the bounds given in Theorem \ref{thm:fuzzy_sorting} using such techniques.

\begin{lemma}
Given $\lambda =  \rho \exp(\iota 2 \pi \phi) \in \sigma(B)$, with $\rho = 1-\epsilon$, and $B \bx = \lambda \bx$, where 
\beq
\label{eqn:rhoandtfuzz}
\lambda = \rho \exp(\iota 2 \pi \phi)  =   \frac{1}{d_i}\sum_{j \in \cN_i} \frac{p_j}{p_i} \exp(\iota 2 \pi (t_j - t_i))
\eeq
for $x_i  = p_i \exp(\iota 2 \pi t_i) \neq 0$ and $\cN_i$ the set of nodes in the out-neighborhood of node $i$, then $p_j$ decays no faster than $(1-\epsilon)$-slowly as we move away from $i$ along edges in $E$.
\label{lemma:fuzzy1}
\end{lemma}

\begin{proof}
Let $i$ be a vertex such that $|x_i| \geq |x_j|$, or  $p_i \geq p_j$, for all $j \neq i$.
Applying (\ref{eqn:rho}) from the proof of Theorem \ref{thm:sorting} gives: 
$$
(1- \epsilon) \leq \max_{j \in \cN_i} \frac{p_j}{p_i},
$$
or $p_j \geq (1-\epsilon) p_i$ for all $j \in \cN_i$.   Now, given a fixed $j \in \cN_i$, consider (\ref{eqn:rho}) applied centered at vertex $j$: 
$$
(1- \epsilon) \leq \frac{1}{d_j}\sum_{k \in \cN_j} \frac{p_k}{p_j} 
$$
If $i$ is not in the out-neighborhood of $j$, following the same method as above, it is easy to see that $p_k \geq (1-\epsilon)p_j \geq (1-\epsilon)^2 p_i$.  This can be continued as we step farther and father away from $i$.

If $i \in \cN_j$, a similar inequality holds.  The above equation can be rewritten as:
$$ 
(1- \epsilon) \leq
\frac{1}{d_j} 
\left( 
\frac{p_i}{p_j} + \sum_{k \in \cN_j \setminus \{i\}} \frac{p_k}{p_j} 
\right) \leq 
\frac{1}{d_j(1-\epsilon)} + \frac{d_j-1}{d_j} \max_{k \in \cN_j\setminus \{i\}} \frac{p_k}{p_j} 
$$
From here, we simplify to see:

$$
\max_{k \in \cN_j\setminus \{i\}} \frac{p_k}{p_j} \geq 
\frac{d_j (1-\epsilon) - (1-\epsilon)^{-1}}{d_j-1}  = \frac{d_j - \epsilon d_j -1 - \epsilon (1-\epsilon)^{-1}}{d_j - 1} = 1 - \epsilon \left( \frac{d_j + (1-\epsilon)^{-1}}{d_j-1} \right).
$$
Now, $p_k \geq \left(1 - \epsilon \left( \frac{d_j + (1-\epsilon)^{-1}}{d_j-1} \right)\right)p_j \geq \left(1 - \epsilon \left( \frac{d_j + (1-\epsilon)^{-1}}{d_j-1} \right)\right)(1-\epsilon)p_i$.

As we step farther and father away from node $i$, at each step to node $j$, $p_j$ decays either by a faction of $(1-\epsilon)$ or by $\left(1 - \epsilon \left( \frac{d_j + (1-\epsilon)^{-1}}{d_j-1} \right)\right)$ and the claim holds.\\

\end{proof}

\begin{lemma}
Given the conditions of Lemma \ref{lemma:fuzzy1}, then for any (directed) edge $(i, j) \in E(G)$ the phase change $2 \pi (t_j - t_i)$ differs from $\phi$ by no more than
$$
\cos^{-1} \left( \frac{0.0199-1.98d_i(1-\epsilon)}{ 2 d_i (1 - \epsilon)} \right).
$$
\label{lemma:fuzzy2}
\end{lemma}

\begin{proof}
Recall equation \ref{eqn:rhoandt} from the proof of Theorem \ref{thm:sorting}:
$$
\rho \exp(\iota 2 \pi \phi)  =   \frac{1}{d_i}\sum_{j \in \cN_i} \frac{p_j}{p_i} \exp(\iota 2 \pi (t_j - t_i))
$$
where $p_i \geq p_j$ for all $j\neq i$ and $\cN_i$ is the set of nodes in the out-neighborhood of node $i$.  This can be rewritten as
$$
d_i \rho  =   \sum_{j \in \cN_i} \frac{p_j}{p_i} \exp(\iota 2 \pi (t_j - t_i - \phi)).
$$

 Plugging in $\rho = 1-\epsilon$ and $p_i = 1$ into the above, we see 
$$
d_i(1- \epsilon) \leq \sum_{j \in \cN_i} p_j \exp(\iota 2 \pi (t_j - t_i) -\phi)
$$

As $p_j \in [(1-\epsilon),1]$ for all $j$, the above is geometrically equivalent to choosing $d_i$ vectors with lengths in $[(1-\epsilon), 1]$ which sum to a vector of length greater than or equal to $d_i(1-\epsilon)$ at an angle of 0.   The maximum difference between the angle of any of these vectors can differ from the angle of the summation vector, 0, is given by letting one vector have unity length with a large deviation from 0 and taking the other $d_i-1$ vectors  to have the same, smaller deviation so that they close the triangle formed by the first vector and the vector of length $d_i(1-\epsilon)$.  The total length of these vectors should be chosen to be close to but less than  $d_i(1-\epsilon) + 1$, so that the triangle inequality holds.  Here, we use the length $d_i(1-\epsilon) + 0.99$.   This produces a triangle with sides $1$, $d_i(1-\epsilon) + 0.99$, and $d_i(1 - \epsilon)$.   The maximum deviation from $\phi$ is given by the angle opposite the side of length $d_i(1-\epsilon) + 0.99$ and can be solved for via the Law of Cosines:
$$
C = \cos^{-1} \left(  \frac{a^2 + b^2 - c^2}{2 a b}  \right) 
$$
Plugging in the appropriate values for $a, b$, and $c$ and simplifying, the fraction inside the inverse cosine becomes:
$$
\frac{1+d_i^2 (1 - \epsilon)^2 - (d_i(1-\epsilon)+0.99)^2}{2 d_i (1 - \epsilon)} = \frac{0.0199-1.98d_i(1-\epsilon)}{ 2 d_i (1 - \epsilon)} 
$$
This proves the claim of Lemma \ref{lemma:fuzzy2}. \\
\end{proof}

\begin{theorem}
Let $\hat{G}$ and $\epsilon$ be defined as in Theorem \ref{thm:fuzzy_sorting}.  Then, there exist eigenpairs of $\hat{B}$, $(\lambda_o, \bx)$ and $(\overline{\lambda_o}, \overline{\bx})$, which can be used to define a mapping into $\mathbb{R}^2$ such that each node $j$ in the highly 3-cyclic area of the network is mapped into a circle of radius $r$ where
$$
r ^2 =1 +\gamma^{2k} - 2\gamma^k \left( \frac{0.0199-1.98d_{max}(1-\epsilon)}{ 2 d_{max} (1 - \epsilon)} \right) 
$$
around vectors of length 1 at angles of $\frac{2\pi}{3}, \pi$, $\frac{4 \pi}{3}$, or $2\pi$ where $d_{max}$ is the maximum degree in the highly 3-cyclic region of the network, $k$ is the number of steps on the shortest path between node $i$ associated with $|x_i| \geq |x_l|$ for all $l$ and node $j$ and $\gamma = 1 - \epsilon \left( \frac{d_{max} + (1-\epsilon)^{-1}}{d_{max}-1} \right)$.
\label{thm:fuzzy_mapping}
\end{theorem}

\begin{proof}
By Theorem \ref{thm:fuzzy_sorting}, there exists $\lambda \in \sigma(B)$ such that $\left| \lambda - \theta_{1,3}\right| < \epsilon$.  Let this eigenvalue be $\lambda_o$ with right eigenvector $\bx$ scaled so that $\max_{i} |x_i| = 1$.  Then, by Lemmas \ref{lemma:fuzzy1}, given node $j$ associated with entry $x_j = p_j\exp(i2\pi t_j)$, $p_j \geq \left(1 - \epsilon \left( \frac{d_j + (1-\epsilon)^{-1}}{d_j-1} \right)\right)^{s}(1-\epsilon)^{k-s}$ for some $0 \leq s \leq k-1$, so $p_j \geq \left(1 - \epsilon \left( \frac{d_j + (1-\epsilon)^{-1}}{d_j-1} \right)\right)^k = \gamma^k$.    By Lemma \ref{lemma:fuzzy2}, the maximum deviation from the angle $\phi$ is given by $\delta = \cos^{-1} \left( \frac{0.0199-1.98d_i(1-\epsilon)}{ 2 d_i (1 - \epsilon)} \right)$.

Then, $x_j$ can be plotted into $\mathbb{R}^2$ to form a vector of length minimum, $\left(1 - \epsilon \left( \frac{d_j - (1-\epsilon)^{-1}}{d_j-1} \right)\right)^k$, with an angle of at most $\delta$ between it and a vector of length 1 at an angle of $\frac{2\pi}{3}, \frac{4\pi}{3}$, or $2\pi$.  By the Law of Cosines, the distance $r$ between the tips of the two vectors is given by
$$
r^2 = 1^2 + \left(1 - \epsilon \left( \frac{d_j + (1-\epsilon)^{-1}}{d_j-1} \right)\right)^{2k} 
- 2 \left(1 - \epsilon \left( \frac{d_j + (1-\epsilon)^{-1}}{d_j-1} \right)\right)^k \cos(\delta) .
$$
Replacing $d_i$ and $d_j$ with $d_{max}$, plugging in for $\delta$, and simplifying completes the proof. \\
\end{proof}

Theorem \ref{thm:fuzzy_mapping} provides bounds on how well grouped the nodes in $V_0, V_1,$ and $V_2$ will be when embedded into $\mathbb{R}^2$ using the methodology from Lemma \ref{lemma:mapping}. When both $\epsilon$ and $d_{max}$ are small, the radius, $r$, of the circle into which the nodes are mapped is close to 0 for nodes within one step of node $i$ and grows slowly with respect to $k$.  In graphs $G$ where the highly 3-cyclic region of the graph is small and the three groups of nodes $V_0, V_1$, and $V_2$ have many connections between them, it can be expected that most nodes in $G$ are within three steps of node $i$ and, thus, will be mapped to three highly clustered areas in $\mathbb{R}^2$.  This identifies the nodes in the three groups which compromise the highly 3-cyclic region of the network.  However, even in networks where $\epsilon$ is larger, the network has high degrees, or there are many nodes in the highly 3-cyclic region which are more than three steps of $i$, experimentally almost always at least one node from each of $V_0$, $V_1$, and $V_2$ will be well-separated from nodes that are not in the 3-cyclic region.  This can be seen in the examples shown in Section \ref{sec:experiments}.

\section{Experiments}
\label{sec:experiments}
In this section, we show the effectiveness of the above methods for finding highly 3- and 4-cyclic regions in a variety of networks, both generated using a stochastic block model and from a variety of real world applications.  In the following experiments, we restrict ourselves to the examination of smaller networks, so that all of the eigenvalues of the row stochastic adjacency matrices can be computed explicitly.  However, in applications with larger datasets, eigenvector approximation methods can be used (see \cite{MaSo96}, among others).  In the experiments below, we calculated all of $\sigma(B)$ and the associated eigenvectors with MATLAB's {\tt eig()} function, which uses the QZ-algorithm for non-symmetric matrices.   All the eigen-residuals have norm less than 1e-14.

The technique for finding highly-cyclic structure we present in this work makes use of embeddings from complex-valued eigenvectors associated with particular complex-valued eigenvalues  of the {\em row-stochastic propogator}, $B$.   For highly 3-cyclic structure, the eigenvector associated with the eigenvalue closest to $\theta_{1,3} = \exp( \iota 2 \pi / 3)$ provides indication of the desired structure.

\subsection{Stochastic Block Models}
\label{sec:experiments1}
Initially, we examine the ability of our proposed methods to identify highly cyclic structures in models with a considerable amount of ground truth.  To begin, we examine the network described in Example \ref{ex0}.   See the left side of Figure~\ref{fig:ex0c} for a plot of the spectrum for Example~\ref{ex0}, $\numext=2$.   The left eigenvector $\by$ of $D^{-1} A$ associated with $\theta_{1,3}$ is complex-valued, and we have a 2d embedding corresponding to the real and imaginary parts of $\by$.   For each vertex $i$ we have the spectral coordinate $(\mbox{Re}\, y_i, \,\mbox{Im}\, y_i)$.   The middle of Figure~\ref{fig:ex0c} shows this embedding for t Example~\ref{ex0}, $\numext=2$, and the right plot shows the similar embedding for the corresponding right eigenvector.   Spatial clustering (such as dbscan \cite{EsKrSaXu96}) easily picks out the classes consisting of non-overlapping 3-cyclic (red/green/blue triangles), overlapping 3-cyclic (red/blue squares), and overlapping classical (magenta/ yellow squares).   All non-overlapping classical community structure is mapped near the origin and clustered together.   The plot of the spectrum and the 2d embeddings do not change qualitatively when more external community structure is added (we tested the $\numext=8$ and 14 cases, and the additional external communities were embedded near the origin without significant changes to the coordinates associated with the 3-cyclic structure).   2d spatial clustering precisely detects the highly-cyclic structure in all cases.

\begin{figure}
\centering
\includegraphics[width = 2.1in]{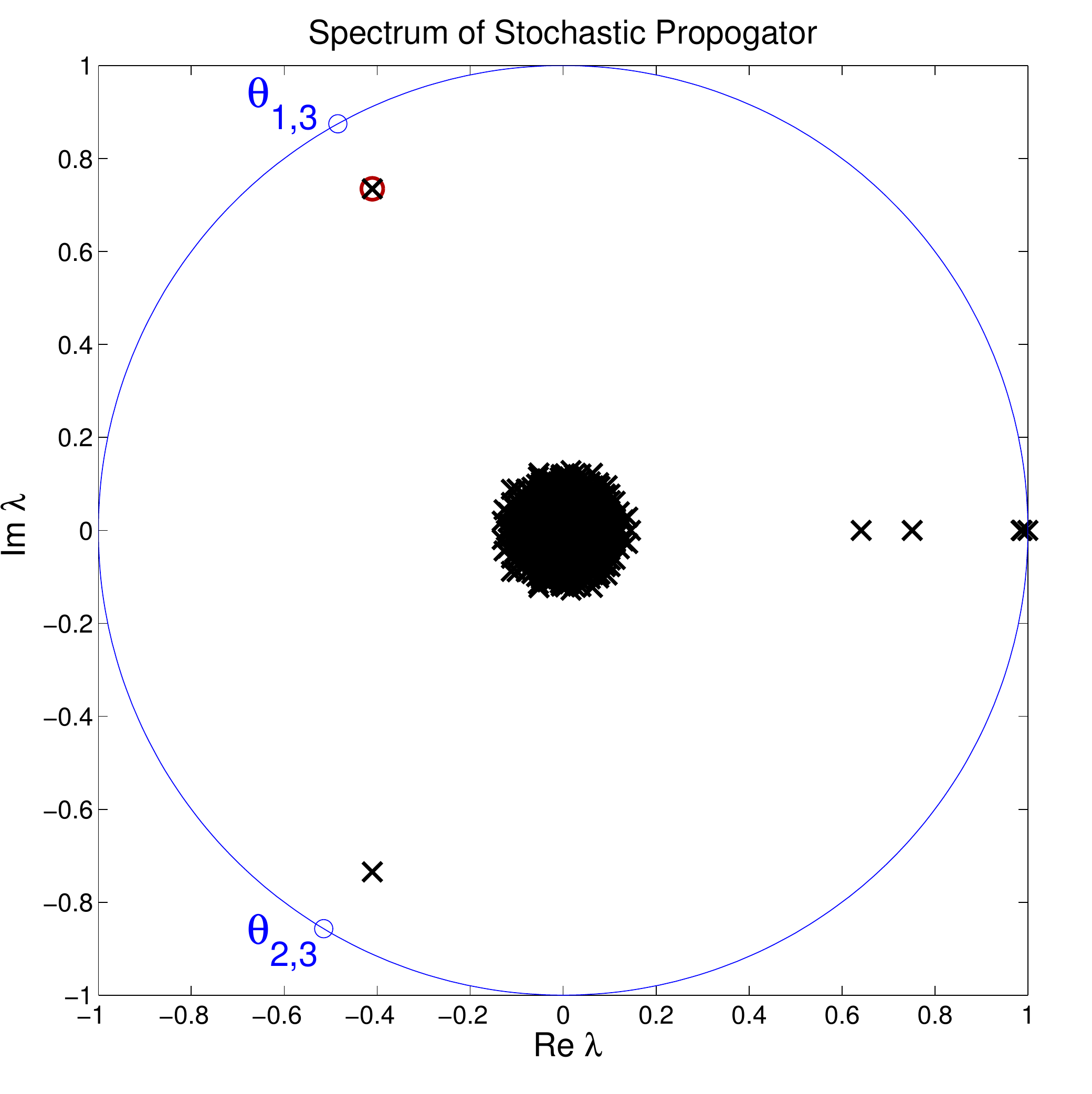}
\includegraphics[width = 2.1in]{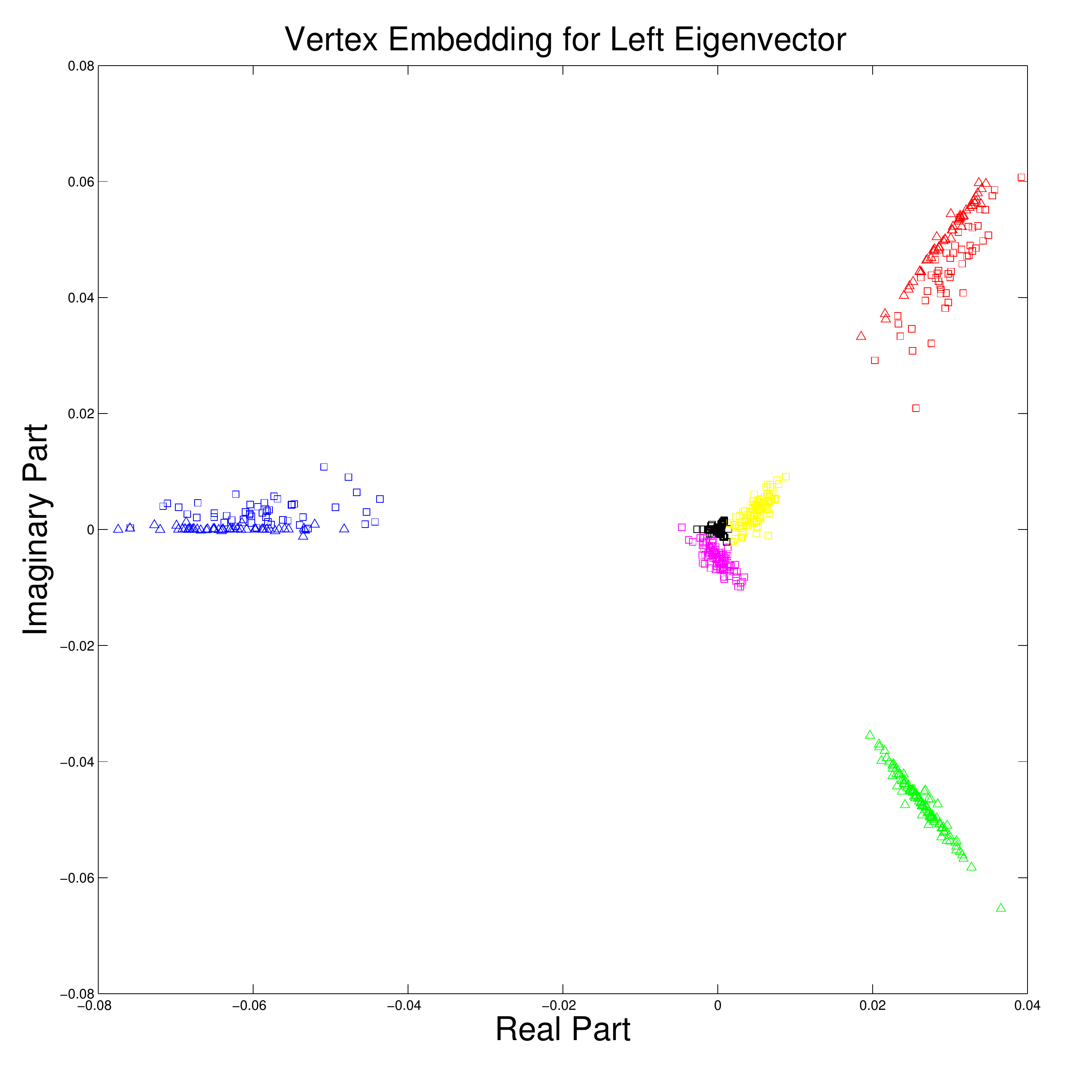}
\includegraphics[width = 2.1in]{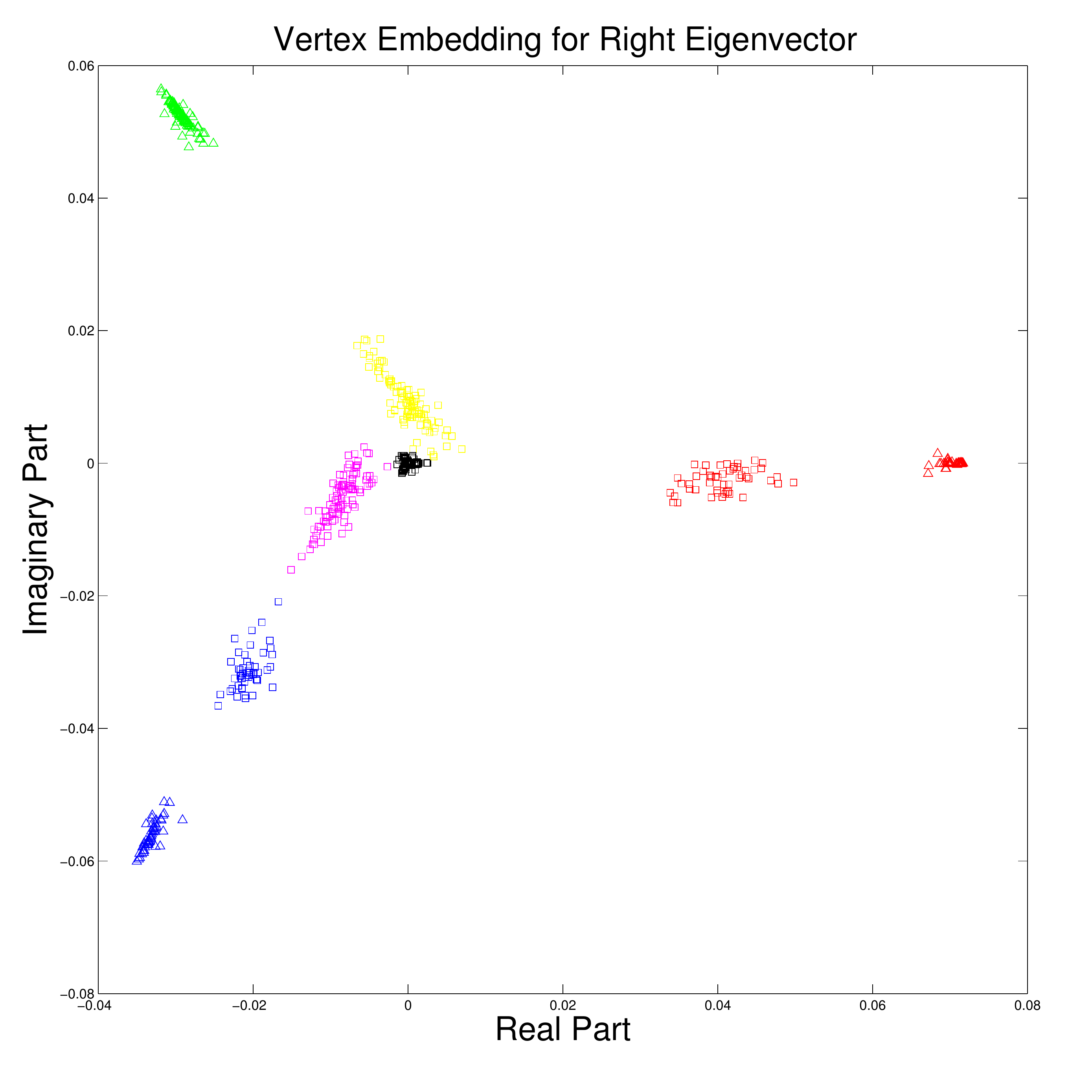} 
\caption{On the left the black x's mark the complex-valued spectrum of $D^{-1} A$ for Example~\ref{ex0}, $\numext=2$.   The eigenvalue closest to $\theta_{1,3} = e^{\iota 2 \pi / 3}$ is of interest (upper-left, circled in dark red).   On the right we plot the real and imaginary parts of the eigenvector associated with this eigenvalue.    The color scheme for these embeddings match that of Figure~\ref{fig:ex0b}.}
\label{fig:ex0c}
\end{figure}

This example suggests that a planar embedding from a single complex eigenpair may be more robustly useful for detecting highly-cyclic structure than a high-dimensional SVD-based embedding.   The eigenvector approach does not require higher-dimensional embeddings for graphs with larger number of non-cyclic structures.   Spatial clustering is greatly simplified, and no search for a useful projection is necessary.   Follow-on analysis for grouping the classes into cyclic structures is also greatly simplified.   The embedding from the left eigenvector seem to be more useful for separating vertices in the 3-cyclic structure from the communities they overlap, whereas that from the right eigenvector does a better job of breaking up the 3-cyclic structure in to the vertices that are internal to the 3-cyclic structure and those that overlap with some classical community structure (see the middle and right of Figure~\ref{fig:ex0c}).   In this work, we focus our analysis on the embedding associated with the left eigenvector, but remark that extending this analysis to the right eigenvector and understanding the interplay of information from both embeddings is an exciting next step. 

Next, we use a stochastic block model to create a network with a blend of non-overlapping non-cyclic, highly 2-cyclic, 3-cyclic, and 4-cyclic substructure.   Specifically, we build a synthetic digraph using a stochastic block model generator that has one classical random digraph community, one two-cyclic community, one three cyclic community, and one four-cyclic community all containing the same number of vertices.  This leads to a network with 10 mutually exclusive groups of vertices, $V_k$, $k=0, ..., 9$, with sizes 
$$
|V_0| = 120, \,
|V_1| = |V_2| = 60, \,
|V_3| = \cdots = |V_5| = 40, \, \mbox{ and } \,
|V_6| = \cdots = |V_9| = 30.
$$  Thus, $|V_0| = |V_1 \cup V_2| = |V_3 \cup \cdots \cup V_5| = |V_6 \cup \cdots \cup V_9| = 120$.

The existence of any edge is governed by one of two probabilities, $\rho_{in}$ and $\rho_{out}$ with $\rho_{in} >> \rho_{out}$.  Then, the probability of directed edge $(i,j) \in E$, $P_{ij}$ is dependent only on the group memberships of $j$ and $i$ where:
$$P_{ij}=\left\{\begin{array}{ll}
\rho_{in},& \textnormal{ if } i,j \in V_0,\\
\rho_{in},& \textnormal{ if } i \in V_1 \textnormal{ and } j \in V_2,\\
\rho_{in},& \textnormal{ if } i \in V_2 \textnormal{ and } j \in V_1,\\
\rho_{in},& \textnormal{ if } i \in V_3 \textnormal{ and } j \in V_4,\\
\rho_{in},& \textnormal{ if } i \in V_4 \textnormal{ and } j \in V_5,\\
\rho_{in},& \textnormal{ if } i \in V_5 \textnormal{ and } j \in V_3,\\
\rho_{in},& \textnormal{ if } i \in V_6 \textnormal{ and } j \in V_7,\\
\rho_{in},& \textnormal{ if } i \in V_7 \textnormal{ and } j \in V_8,\\
\rho_{in},& \textnormal{ if } i \in V_8 \textnormal{ and } j \in V_9,\\
\rho_{in},& \textnormal{ if } i \in V_9 \textnormal{ and } j \in V_6,\\
\rho_{out}, & \textnormal{ else.} \\
\end{array}\right .
$$
That is, the probability of any specific edge is given by $\rho_{in}$ if it falls within the dictated community structure and by $\rho_{out}$ otherwise.  In the example displayed below, we set $\rho_{in} = 0.80$ and $\rho_{out} = 0.01$.   We remark that the non-cyclic community $V_0$ has the most internal edges with high probability, implying it is by far the strongest community in the classical sense.   See the left half of Figure \ref{fig:ex1} for the adjacency matrix $A$ associated with a sample from this digraph generator and the spectrum of $B = D^{-1}A$.

\begin{figure}[h]
\centering
\includegraphics[draft = false, width = .53\textwidth]{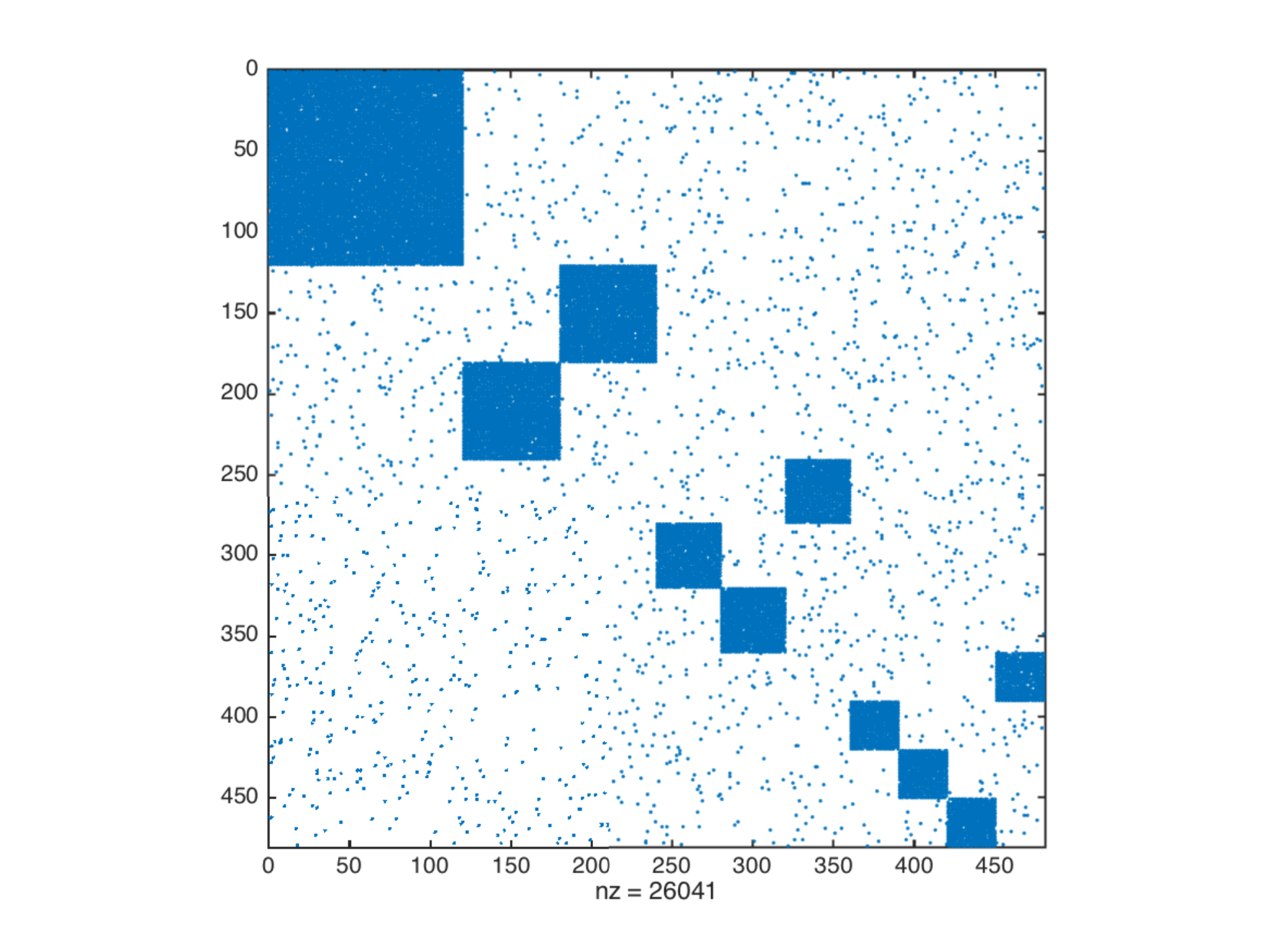} 
\hspace{0.15in}
\includegraphics[draft = false, width = .43\textwidth]{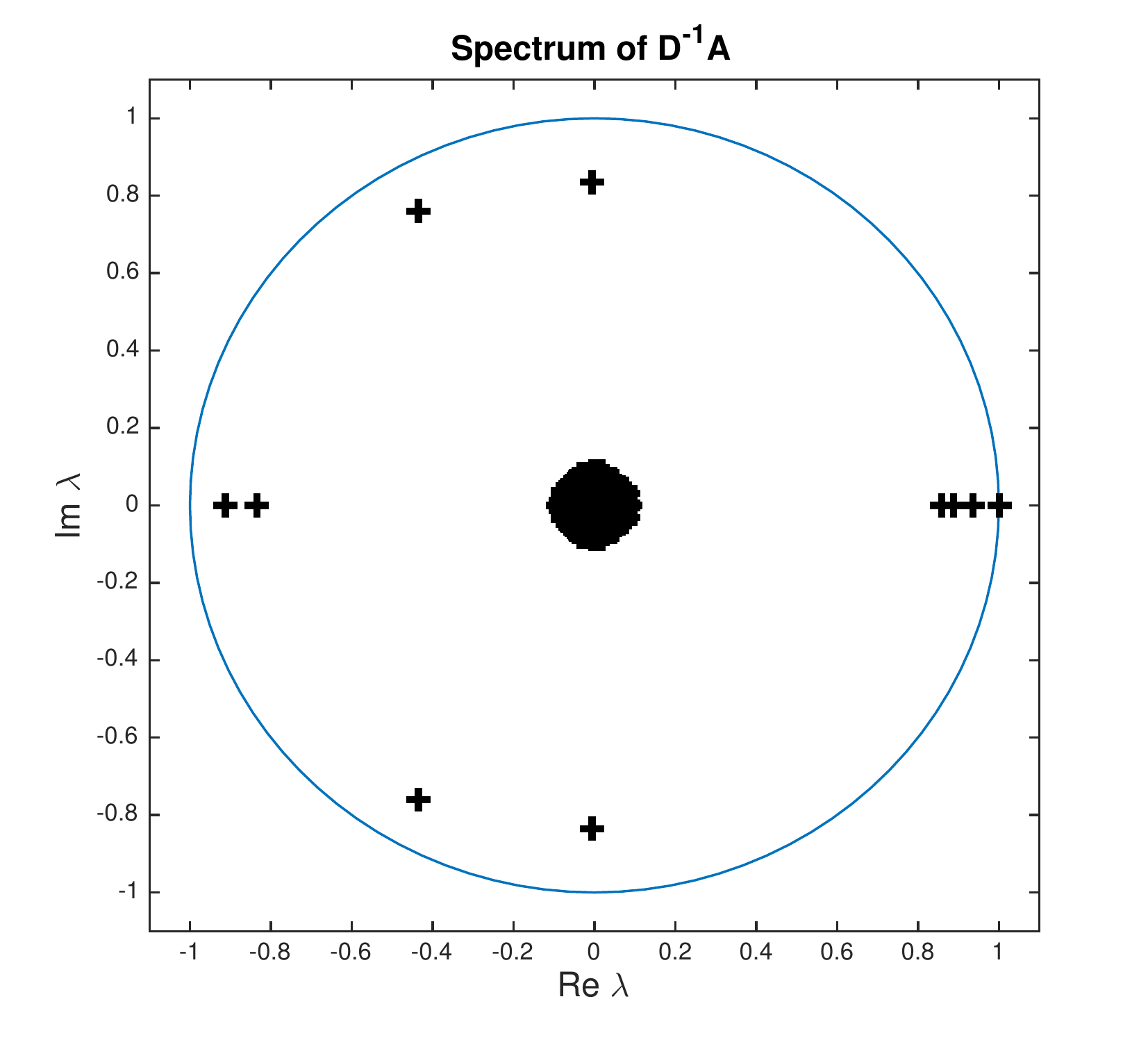} 
\caption{The adjacency matrix (left) and spectrum of row stochastic adjacency matrix (right) of a test problem generated using a stochastic block model.}
\label{fig:ex1}
\end{figure}

We investigate the properties of $\sigma(B)$, focusing on the properties of the eigenspaces associated with the eigenvalues closest to $\theta_{1,3}$ and $\theta_{1,4}$ and verifying that they help us identify which vertices are in highly 3- and 4-cyclic structures, respectively.  A plot of the calculated spectrum in the complex plane can be found on the right half of Figure \ref{fig:ex1}.  The eigenvalue closest to $\theta_{1,3}$ is $\lambda \approx -0.4347 + 0.7535\iota$ (so, $|\theta_{1,3} -\lambda | \approx 0.1301$) and that closest to $\theta_{1,4}$ is $\lambda \approx 0.0010 + 0.8291\iota$ ($|\theta_{1,4} - \lambda | \approx 0.1709$)).  The bounds presented in Theorem \ref{thm:fuzzy_sorting} predict that $|\theta_{1,3} -\lambda| \leq 0.4001$ based on the size of the 3-cyclic region at 120 vertices, the expected value of $d_i = 32$ based on $\rho_{in}$ and the expected value of $\hat{d}_i = 36.4$.  In this example, the bounds on $|\theta_{1,3} -\lambda|$ are not particularly tight, but if they were the eigenvalue in question would still be well separated from the cluster near zero and the 3-cyclic region would be identifiable.

\begin{figure}[h]
\centering
\includegraphics[draft = false, width = 0.48\textwidth]{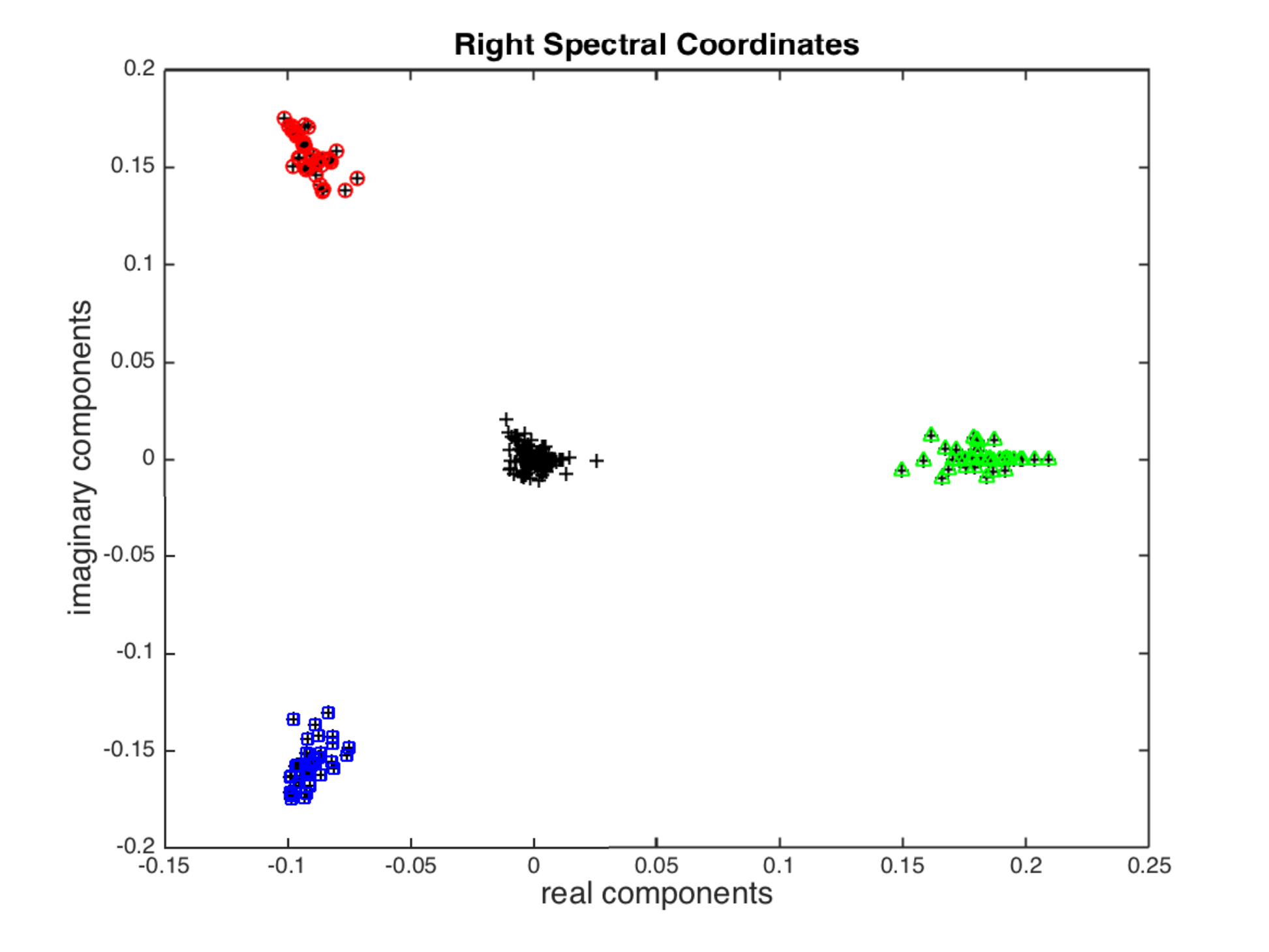}
\hspace{0.15in}
\includegraphics[draft = false, width = 0.48\textwidth]{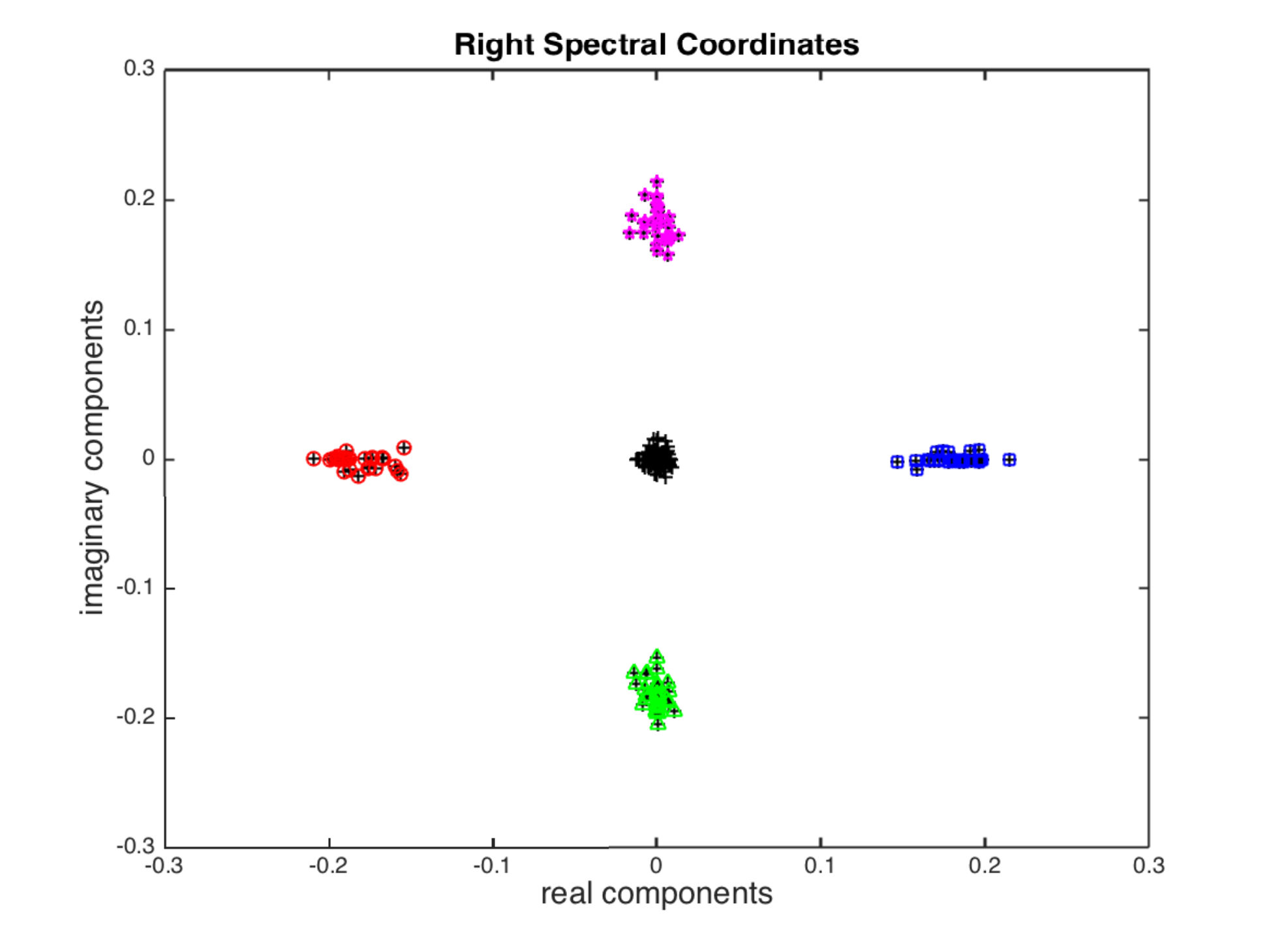} 
\caption{Coordinate embeddings using the right eigenvectors associated with $\lambda \approx \theta_{1,3}$ (left) and $\lambda \approx \theta_{1,4}$ (right) of the synthetic digraph whose adjacency matrix is shown in Figure \ref{fig:ex1}. Black nodes were unclustered in these embeddings.}
\label{fig:ex1scs}
\end{figure}

In Figure \ref{fig:ex1scs}, we embed the nodes of our generated network into $\mathbb{R}^2$ using the right eigenvectors associated with $\lambda \approx -0.4347 + 0.7535\iota \approx \theta_{1,3}$ on the left and $\lambda \approx 0.0010 + 0.8291\iota \approx \theta_{1,4}$.  In the embedding on the right, associated with $\lambda \approx \theta_{1,3}$, the nodes in groups $V_3, V_4,$ and $V_5$ are colored red, green, and blue respectively.  On the right, in the embedding associated with $\lambda \approx \theta_{1,4}$, the nodes in groups $V_6, V_7, V_8,$ and $V_9$ are colored red, green, blue, and magenta, respectively.  In both embeddings, nodes in all other groups are colored black.  In this plot, it is easy to see that the eigenvector associated with $\lambda \approx \theta_{1,3}$ perfectly classifies the 3-cyclic structure in $V_3 \cup V_4 \cup V_5$, while essentially ignoring all other nodes in the network.  According to Lemma \ref{lemma:fuzzy1} (which states that the magnitude of the node embeddings for nodes in the 3-cyclic structure decays no faster than $(1-\epsilon)$ slowly, where $\epsilon = |\theta_{1,3} -\lambda |$) the magnitude of the node embeddings for the red, blue, and green nodes should decay no faster than $(1-\epsilon) \approx 0.8699$.  As seen in the embedding on the left of Figure \ref{fig:ex1scs}, the magnitudes of the node embeddings are within this bound (and often decay even more slowly).  Given that the maximum out-degree in the highly 3-cyclic region of the network is 46 and the node with the largest magnitude in the eigenvector associated with $\lambda \approx \theta_{1,3}$ is node 319, Theorem \ref{thm:fuzzy_mapping} does not provide meaningful information, in that it states that nodes within one step of nodes 319 will be embedded into a circle of radius 1.8904 around a vector of length 1 at an angle of $\frac{2\pi}{3}$, which encompasses the total area in which the nodes have been embedded.  However, even though $\epsilon$ and $d_{max}$ are large enough that Theorem \ref{thm:fuzzy_mapping} is not useful, the nodes are still well separated.

Similarly, the information in the eigenvector associated with $\lambda \approx \theta_{1,4}$ perfectly classifies the 4-cyclic structure in $V_6 \cup \ldots \cup V_9$.  This is to be expected as both the 3-cyclic and 4-cyclic communities are well isolated from the rest of the network.  We will see in Section \ref{sec:real_exp}, that in networks where cyclic structure is not as well isolated (which is the case in many real world complex networks) embedding the nodes no longer fully isolates cyclic communities.


\subsection{Real-World Graphs}
\label{sec:real_exp}
Here, we search for highly 3- and 4-cyclic structure in two real-world directed graphs.  The two graphs we consider here, the Stanford CS web graph and the Enron email network, both of which can be found in the University of Florida Sparse Matrix Collection \cite{UFSparse}.  We calculate the largest strongly connected component (SCC) of both networks using the {\tt MatlabBGL} toolbox \cite{MatlabBGL} and performing the subsequent analysis only on the largest SCC.

The Stanford CS web graph is part of the Gleich group in the UF collection.  Here, the nodes are websites in the Stanford CS domain from 2001 and there is an edge $(i,j) \in E$ if website $i$ links to website $j$.  The original network has 9,914 nodes and 36,854 edges.  The largest strongly connected component has 2,759 nodes and 13,895 edges.  The adjacency matrix of the largest SCC can be seen on the left of Figure \ref{fig:real1} and the spectrum of the row stochastic adjacency matrix is displayed on the right.  The closest eigenvalue to $\theta_{1,3}$ is given by $\lambda \approx -0.4671 + 0.8249\iota$, thus $|\theta_{1,3} - \lambda| \approx 0.0527$.  

\begin{figure}[h]
\centering
\includegraphics[width = .52\textwidth]{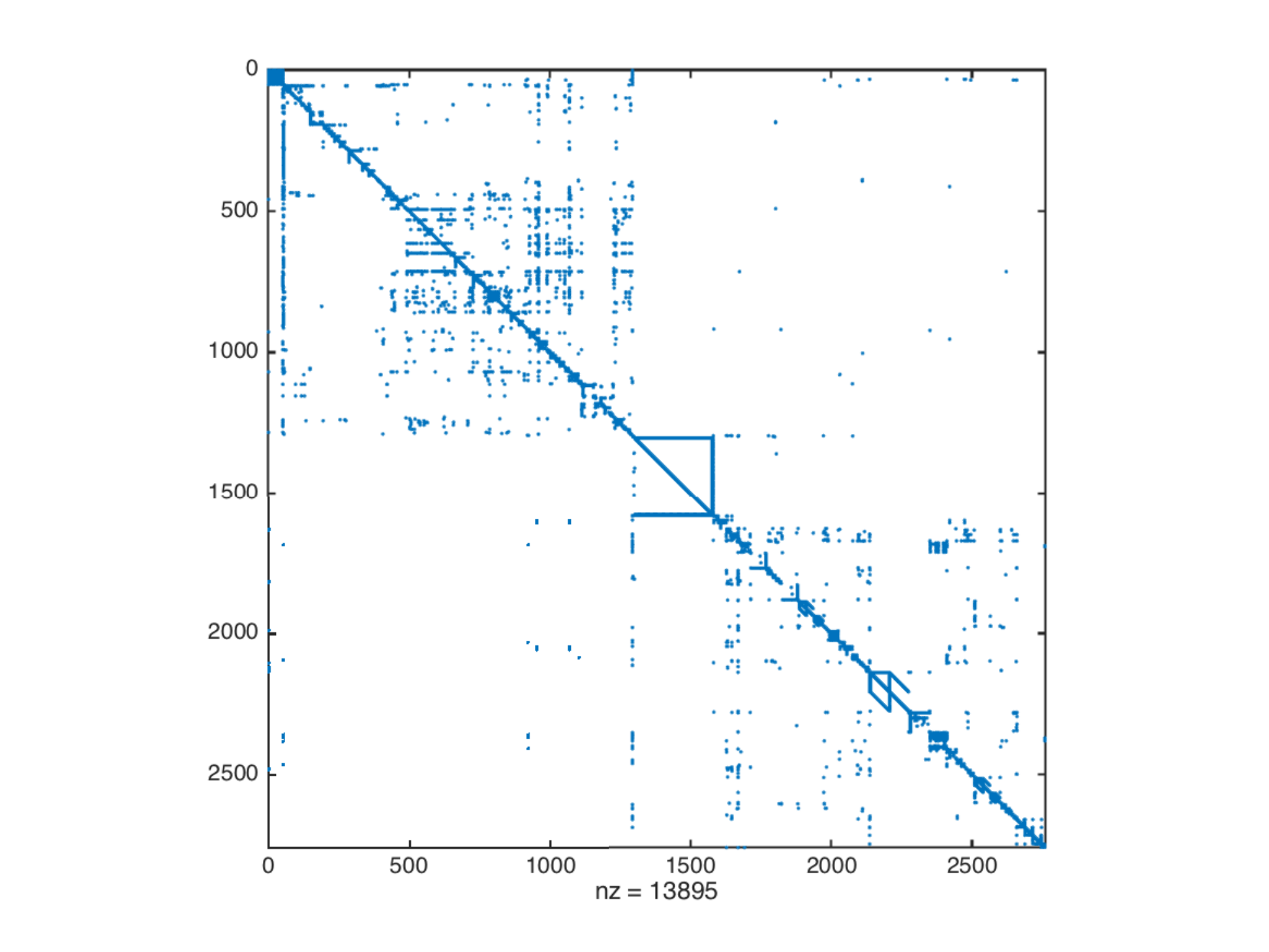} 
\hspace{0.15in}
\includegraphics[width = .44\textwidth]{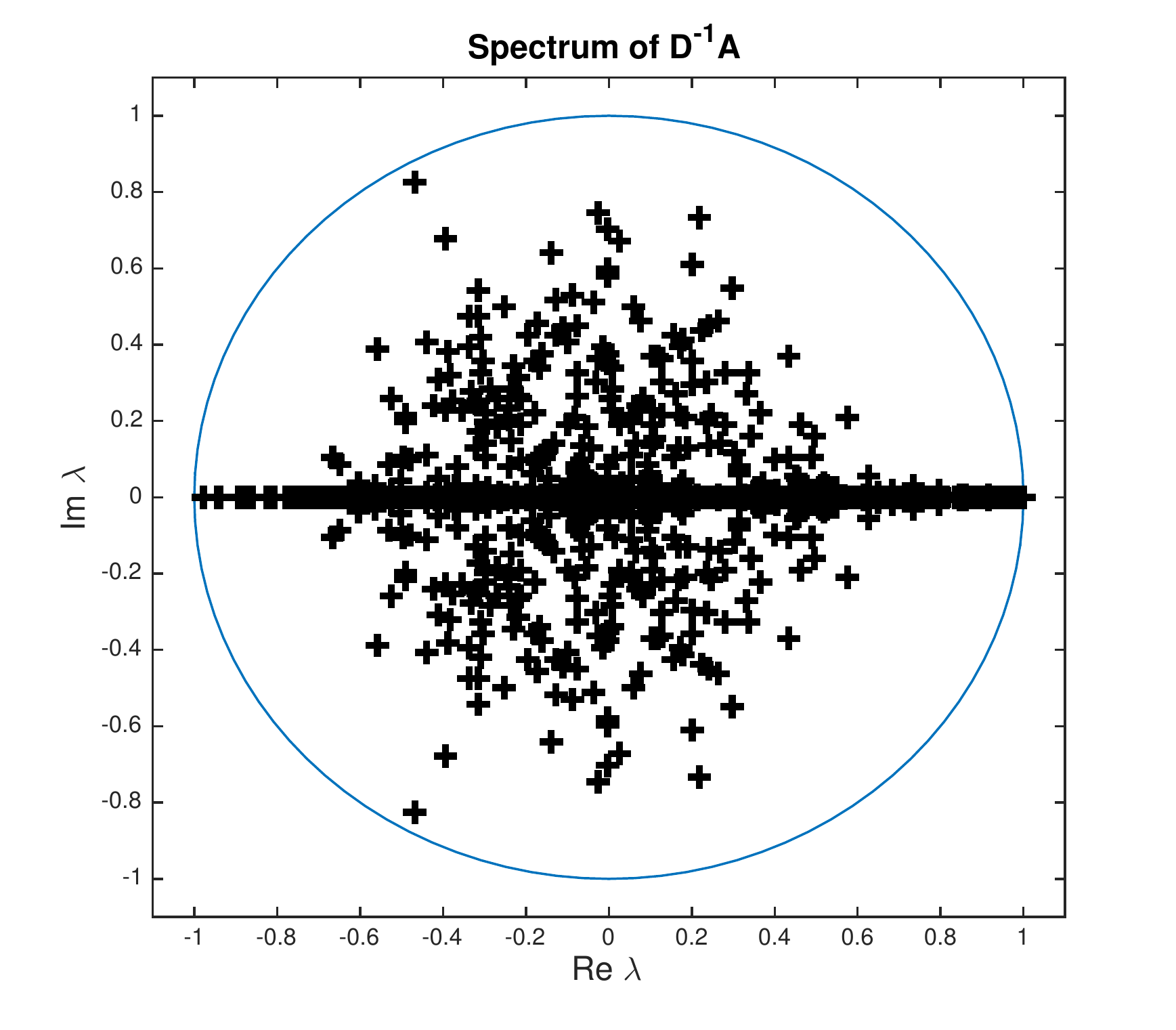} 
\caption{The adjacency matrix (left) and spectrum (right) of the 2001 Stanford webpage network.}
\label{fig:real1}
\end{figure} 

The closeness of $\lambda$ to $\theta_{1,3}$ indicates that there is some highly 3-cyclic structure in the Stanford CS web graph which is well-separated from the rest of the network.  The highest degree in the network is 277 (although, without further analysis it is not clear whether or not the node associated with this degree is in the highly 3-cyclic region of this network).  Using 277 as an upper bound on the maximum degree of a node in the highly 3-cyclic region, Lemma \ref{lemma:fuzzy1} states that the magnitude of the node embeddings will decay no faster than a rate of approximately 0.9469.  The embedding nodes of the Stanford CS web graph into $\mathbb{R}^2$ using the eigenvector associated with $\lambda \approx \theta_{1,3}$ are displayed in Figure \ref{fig:real1_3cyclic}.  Given the large decay rates in the embeddings, all nodes in the 3-cyclic region appear to be well separated.   However, it does not take much connectivity between the 3-cyclic area and the rest of the graph to lead to nodes which are embedded between the 3-cyclic nodes and the rest of the graph, especially when the probability of edges among the 3-cyclic groups is not as high as in the generated networks (this phenomena can also been seen in Figure \ref{fig:ex0c}). Without more information about the exact websites which are involved in this 3-cyclic structure, it is difficult to speculate on an explanation for this 3-cyclic structure.  The complex eigenvalues, however, identify that this 3-cyclic structure exits and provide a starting point for deeper analysis.

\begin{figure}[h]
\centering
\includegraphics[width = .48\textwidth]{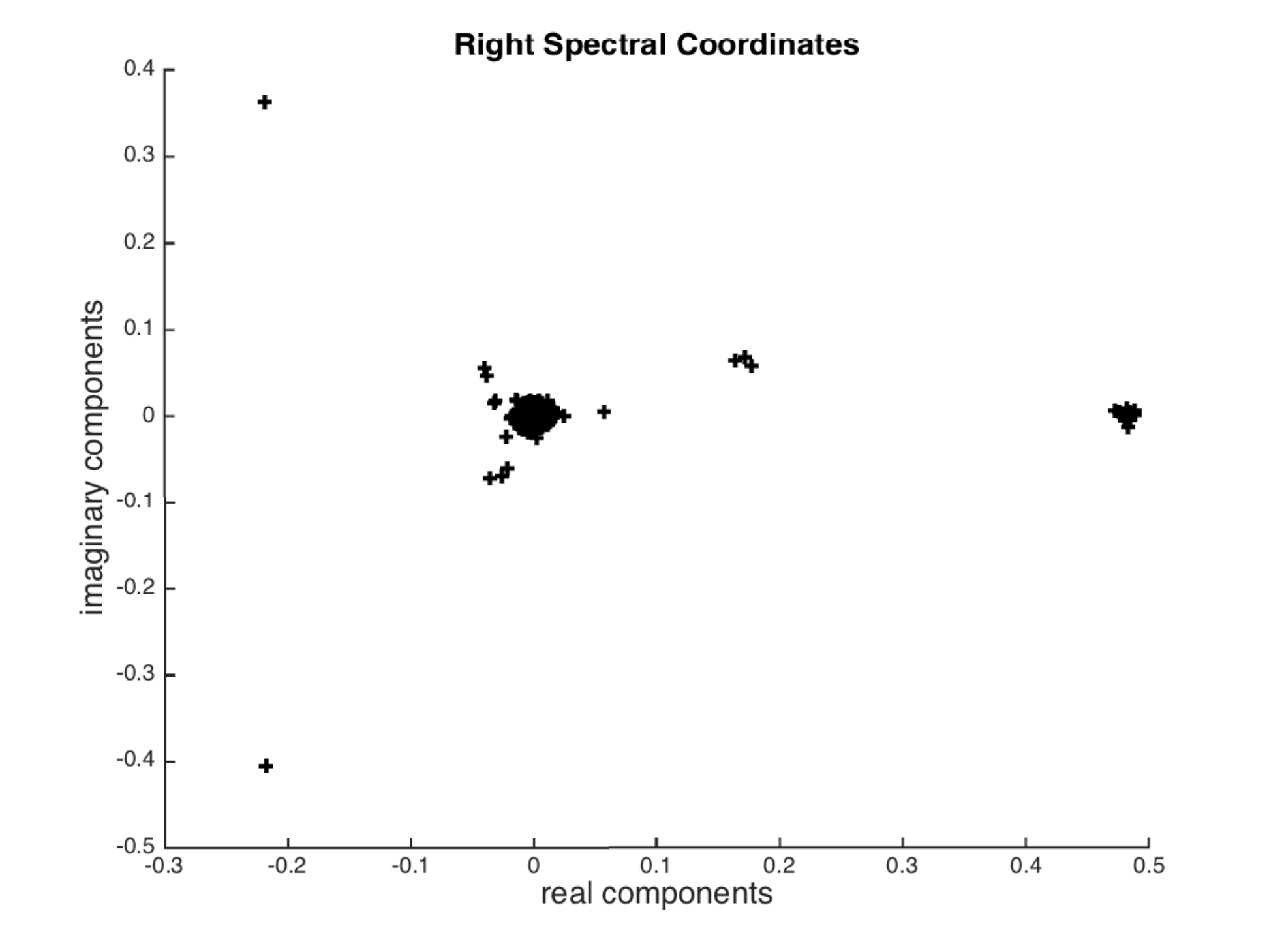} 
\caption{Coordinate embeddings using the right eigenvector associated with $\lambda \approx \theta_{1,3}$ of the 2001 Stanford webpage network.  In this figure, a small amount of Gaussian noise has been added to distinguish nodes in the 3-cyclic region that were embedded on top of each other.}
\label{fig:real1_3cyclic}
\end{figure} 

The embedding nodes of the Stanford CS web graph into $\mathbb{R}^2$ using the eigenvector associated with $\lambda \approx \theta_{1,3}$ are displayed in Figure \ref{fig:real1_3cyclic}.  Here, it is clear that the majority of the nodes in the network are not clearly identified as belonging to the 3-cyclic structure, however there is at least one clearly identified node from each of the 3-cyclic groups.  These can be used as seed nodes in other community detection networks, such as \cite{KlKl14,WhGlDh13} and many others.  Without more information about the exact websites which are involved in this 3-cyclic structure, it is difficult to speculate on an explanation for this 3-cyclic structure.  The complex eigenvalues, however, identify that this 3-cyclic structure exits and provide a starting point for deeper analysis and initial analysis indicates that the structure involves links to and from style files.

The next network we examine is the Enron email network.  The version considered in this paper was provided by the Laboratory for Web Algorithmics (LAW) at the Universita degli Studi di Milano and can be found in the LAW group in the UF collection.  In this network, nodes are email addresses and there is an edge from node $i$ to node $j$ if email address $i$ sent an email to address $j$.  The original network has 69,244 nodes and 276,143 directed edges.  The largest strongly connected component has 8,271 nodes and 147,353 edges.  That is, over half of the edges in the original network are present in the largest SCC even though it contains only about 12\% of the original nodes.  The adjacency matrix of the largest SCC of the Enron email network can be found in the left half of Figure \ref{fig:real2}. The spectrum of the row stochastic adjacency matrix can be found on the right.

\begin{figure}[h]
\centering
\includegraphics[width = .52\textwidth]{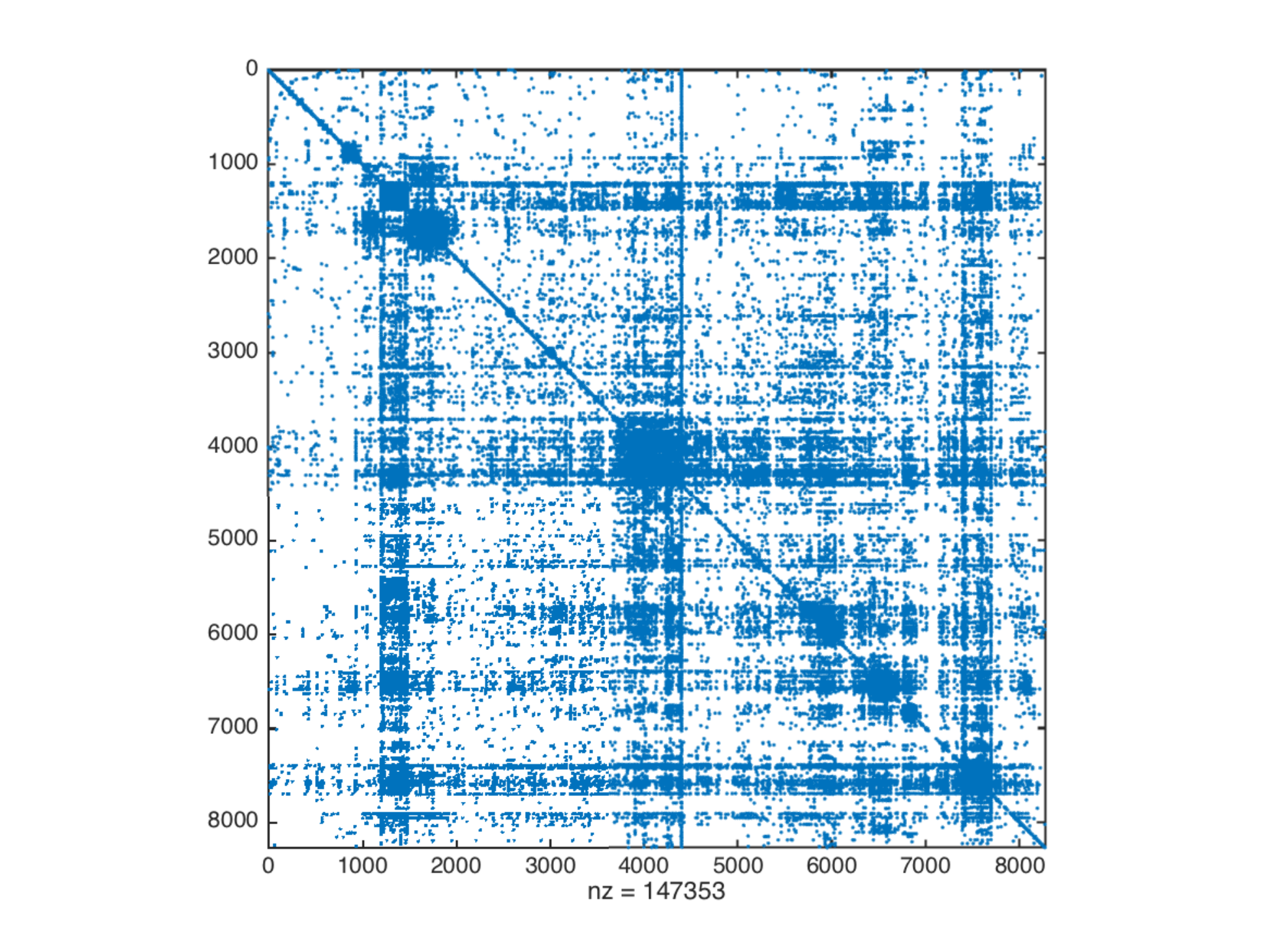} 
\hspace{0.15in}
\includegraphics[width = .42\textwidth]{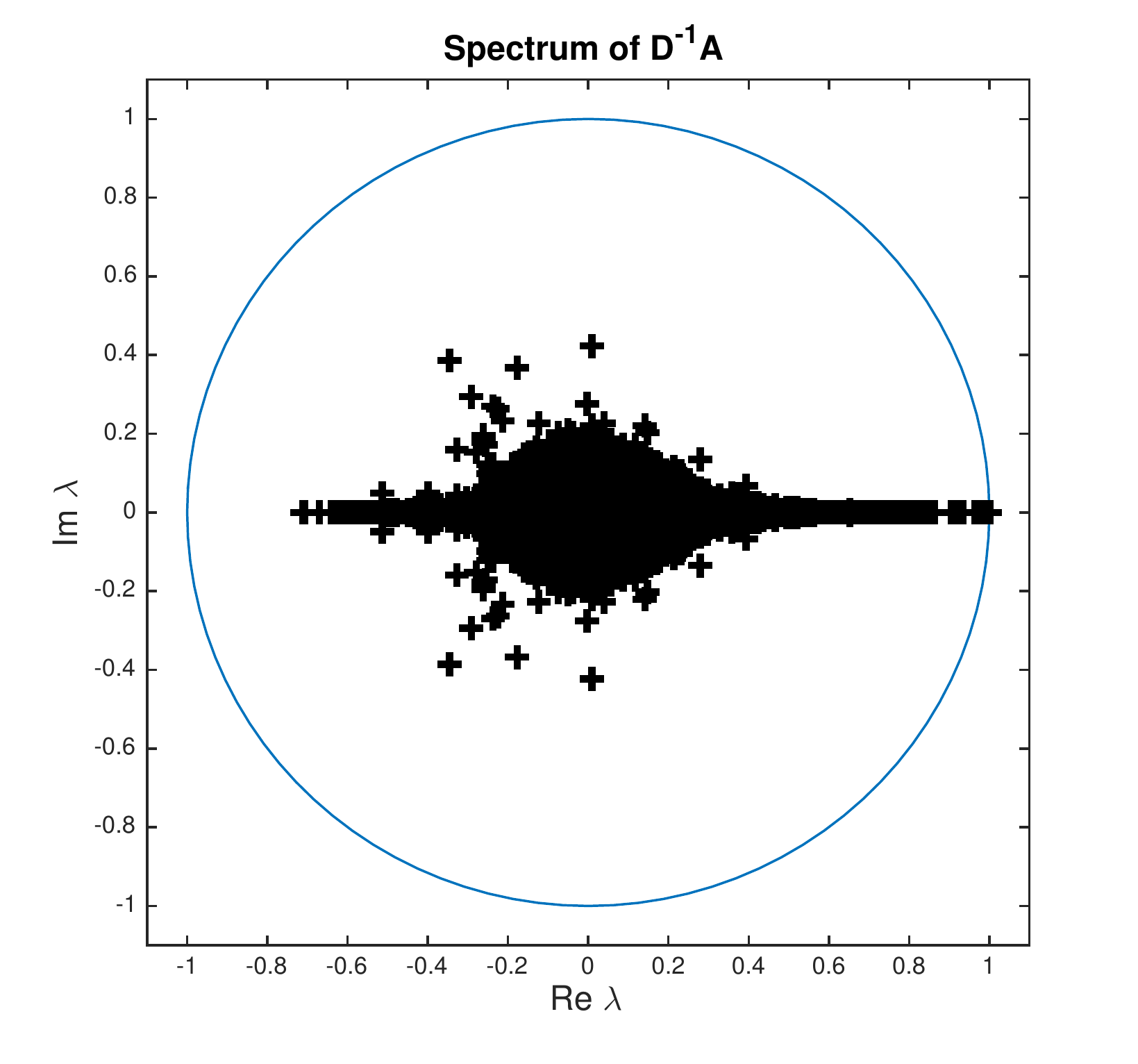} 
\caption{The adjacency matrix (left) and spectrum (right) of the Enron email network.}
\label{fig:real2}
\end{figure} 

The eigenvalues of the Enron email network are concentrated much closer to the origin than in the case of the other networks examined in this paper, indicating that any substructure in the network is very interconnected with the network as a whole.  However, there are still eigenvalues which are separated from the main cluster in the directions of $\theta_{1,3}$ and $\theta_{1,4}$.  The eigenvalue closest to $\theta_{1.3}$ is $\lambda \approx -0.3430 + 0.3866\iota$, which means that $|\theta_{1,3} - \lambda | \approx 0.5045$.  The closest eigenvalue to $\theta_{1,4}$ is $\lambda \approx 0.0104 + 0.4212\iota$, thus $|\theta_{1,4} - \lambda | \approx 0.5789$.  This indicates than any 3- or 4-cyclic substructure is not well-separated from the rest of the network, which is further indicated by the fact that Lemma \ref{lemma:fuzzy1} states that the magnitude of the embeddings of nodes in the highly 3-cyclic region can decay as fast as 0.5045 at each step.  

Even though the highly cyclic substructure is integrated into the Enron email network as a whole, the eigenvectors associated with $\lambda \approx -0.3430 + 0.3866\iota$ and $\lambda \approx 0.0104 + 0.4212\iota$ can still be used to identify one node from each group in the highly 3- or 4-cyclic substructure.  The embeddings of the nodes of the largest SCC into $\mathbb{R}^2$ can be found in Figure \ref{fig:real2_3cyclic}.  The embedding using the eigenvector $\lambda \approx \theta_{1,3}$ is on the left and that using the eigenvector associated with $\lambda \approx \theta_{1.4}$ is on the right.    Here, it is clear that the majority of the nodes in the network are not clearly identified as belonging to the 3-cyclic structure, however there is at least one clearly identified node from each of the 3-cyclic groups.  These can be used as seed nodes in other community detection networks, such as \cite{KlKl14,WhGlDh13} and many others.  In the 4-cyclic case, there is again at least one seed node from each group that is well separated in the embedding.  Combined with the fact that  $\lambda \approx 0.0104 + 0.4212\iota$ is not as close to $\theta_{1,4}$ as $\lambda \approx -0.3430 + 0.3866\iota$ is to $\theta_{1,3}$, this indicates that the 4-cyclic structure in the Enron email network is not as distinctive as the 3-cyclic structure.  And neither of these substructures are as identifiable as the 3-cyclic structure in the Stanford CS web network.  Again, further in-depth analysis is required to determine exactly what is contributing to these structures.

\begin{figure}[h]
\centering
\includegraphics[width = .44\textwidth]{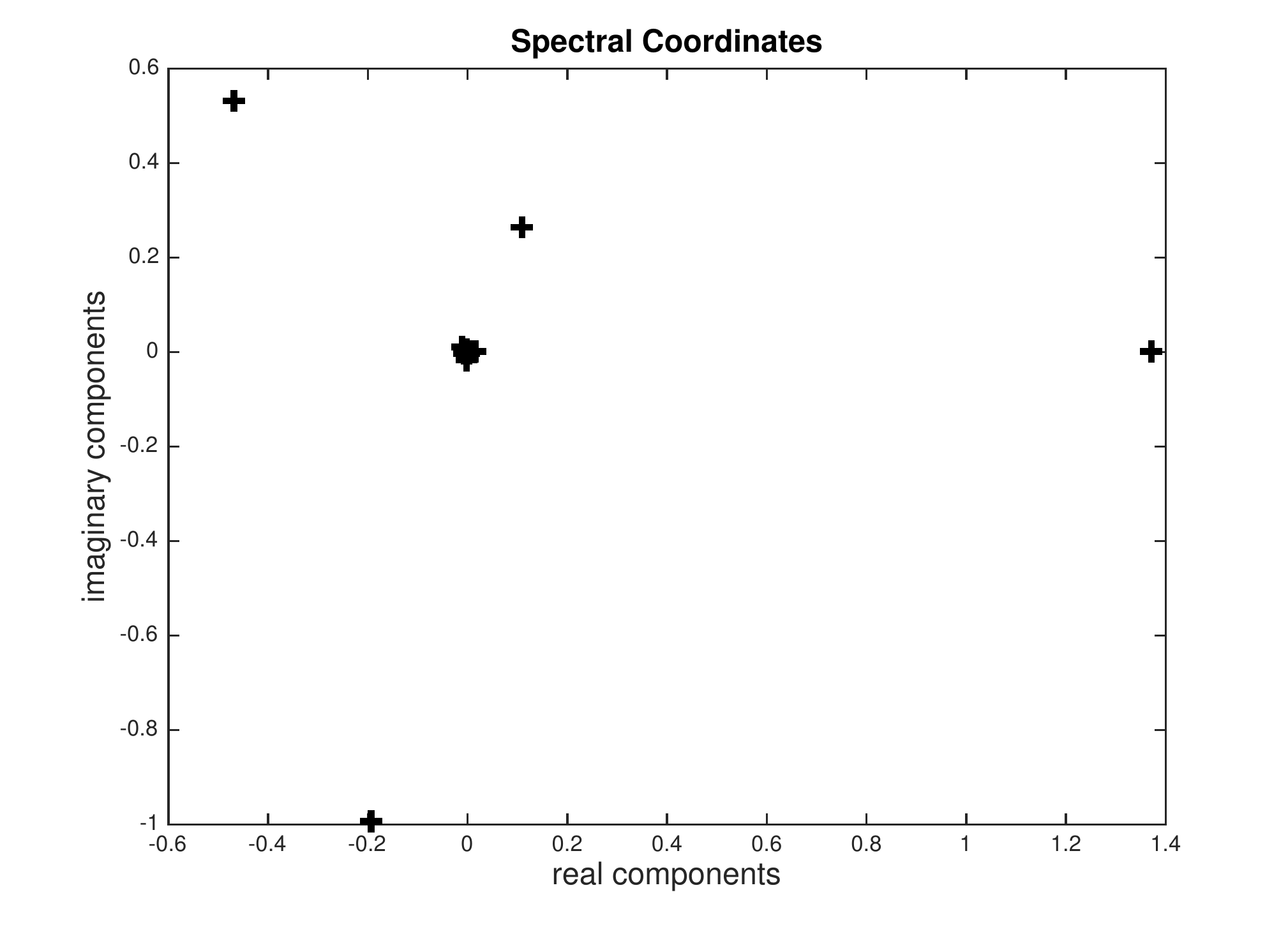} 
\hspace{0.3in}
\includegraphics[width = .44\textwidth]{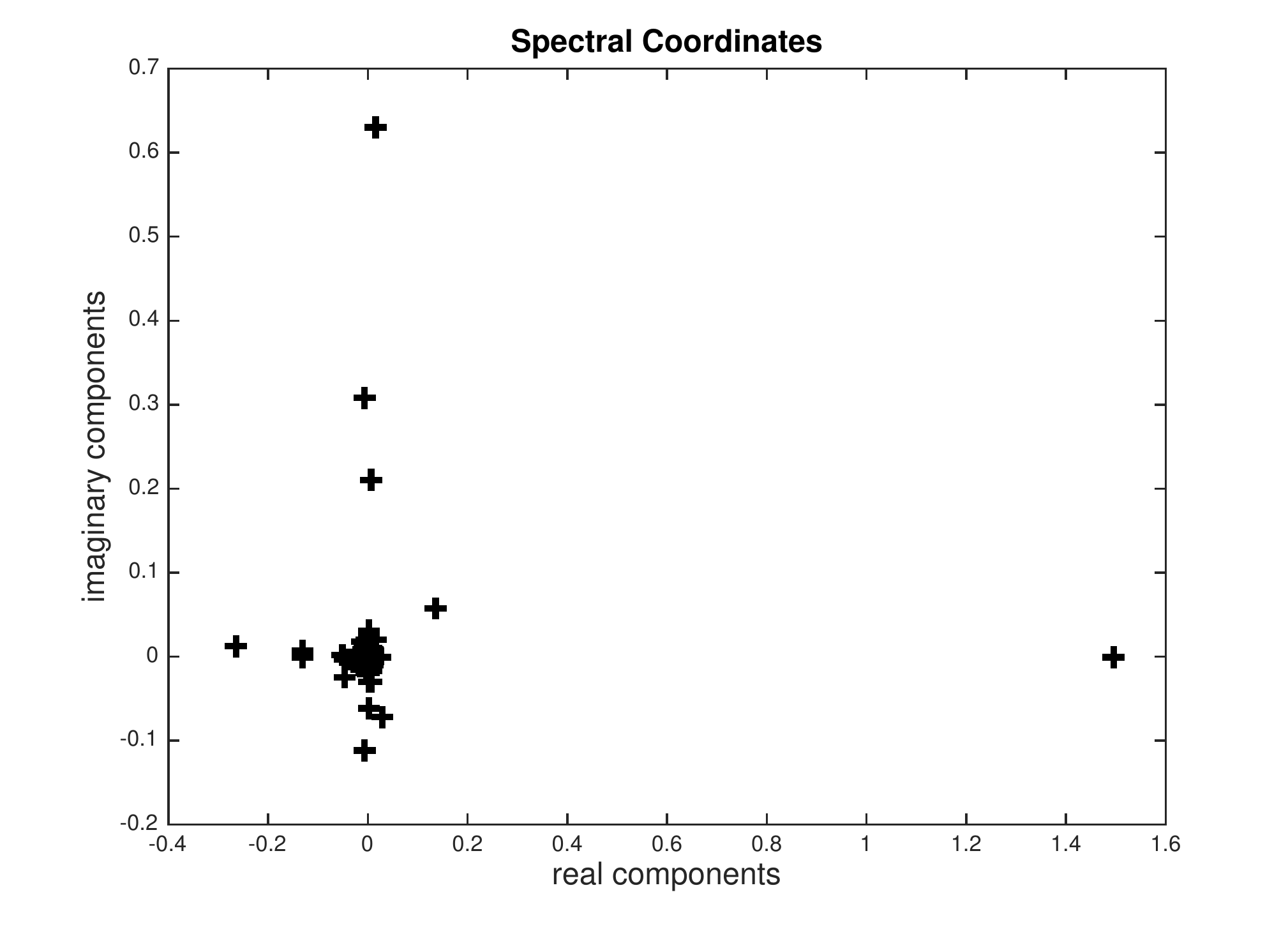} 
\caption{Coordinate embeddings of the right eigenvectors of the Enron email network associated with $\lambda \approx \theta_{1,3}$ (left) and $\lambda \approx \theta_{1,4}$ (right).}
\label{fig:real2_3cyclic}
\end{figure} 

\section{Conclusions and Further Work}
\label{sec:conclusions}

We have studied the relationship between the eigenpairs of row stochastic adjacency matrices of directed networks and the existence of highly cyclic structure in these networks.  In this work, we emphasized networks with both purely and highly 3-cyclic structure, including networks where the highly 3-cyclic region overlapped with dense, non-cyclic communities.  We showed that the existence of eigenvalues at (or near) the imaginary third roots of unity identifies the existence of a 3-cyclic (or highly 3-cyclic) structure in a network and that the eigenvectors associated with these eigenvalues can be used to identify the nodes involved in the individual parts of the 3-cyclic structure.  We additionally demonstrated the effectiveness of these techniques on a variety of generated and real-world networks.  Although our analysis focused on 3-cyclic structure, with slight modifications our methodology can be applied to cycles of any length and we demonstrated the usefulness of eigenvalues near the imaginary fourth roots of unity in identifying highly 4-cyclic structure.  Generally speaking, we suspect that the largest magnitude eigenvalues which are not $\lambda \approx 1$ are likely to provide information regarding general $k$-cyclic structure in a directed network. 

This work is a first step in developing methodologies to identify the existence of communities of varies types of directed structure in complex networks.  Due to the nature of directed edges and the fact that there are, up to isomorphism, seven distinct types of directed triangles, a number of community structures involving edges between three super nodes were not discussed in this paper.  Future work involves extending this methodology to the identification of other types of community structures involving three super nodes.  Another aspect of future work involves improving the bounds in Theorems \ref{thm:fuzzy_sorting} and \ref{thm:fuzzy_mapping} so that they can be effectively be applied to larger 3-cyclic structures.\\

\section*{Funding}
This work was performed under the auspices of the U.S. Department of Energy by Lawrence Livermore National Laboratory under Contract DE-AC52-07NA27344.


\bibliographystyle{siam}  
 \bibliography{3-cyclic_bib}  

\end{document}